\documentclass[12pt,a4paper]{article}%
\usepackage{amsmath}
\usepackage{amsfonts}
\usepackage{amssymb}
\usepackage{graphicx}
\usepackage{hyperref}%

\newtheorem{theorem}{Theorem}
\newtheorem{lemma}[theorem]{Lemma}
\newtheorem{remark}[theorem]{Remark}
\newenvironment{proof}[1][Proof]{\noindent\textbf{#1.} }{\ \rule{0.5em}{0.5em}}
\renewcommand{\theequation}{\thesection.\arabic{equation}}
\numberwithin{equation}{section}
\begin{document}

\title{Structure Functions for small DIS $x$}
\author{Hrachya M. Babujian\thanks{Address: Yerevan Physics Institute,
Alikhanian Brothers 2, Yerevan, 375036 Armenia} \thanks{Address: Beijing
Institute of Mathematical Sciences and Applications, Beijing, China}
\thanks{E-mail: babujian@yerphi.am} ,
Angela Foerster\thanks{Address: Instituto de F\'{\i}sica
da UFRGS, Av. Bento Gon\c{c}alves 9500, Porto Alegre, RS - Brazil}
\thanks{E-mail: angela@if.ufrgs.br} ,\\ and Michael Karowski\thanks{E-mail:
karowski@physik.fu-berlin.de}\\Institut f\"{u}r Theoretische Physik,
Freie Universit\"{a}t Berlin,\\Arnimallee 14, 14195 Berlin, Germany }
\date{\today}
\maketitle

\begin{abstract}
Structure Functions for small DIS (deep inelastic scattering) $x$ for
integrable models are investigated, in particular, for the $O(N)$~$\sigma
$-model and $SU(N)$ chiral Gross-Neveu model, which are asymptotically free.
We get the universal behavior $x^{-1}\ln^{-2}x$ at small Bjorken variable $x$
and confirm a Balog Weisz conjecture. For a the second group of models, the
Sine-Gordon, sinh-Gordon and $Z(N)$, we find power behavior $x^{-\lambda}$.
The special behavior of the structure function for the Sine-Gordon model is
probably crucial for future investigation in 4D QCD.

\end{abstract}

\section{Introduction}

Quantum Chromodynamics (QCD) emerged in the early 1970s as the non-abelian
gauge theory describing the strong interaction, whose hallmark properties -
confinement and asymptotic freedom - were soon established both theoretically
and experimentally. Despite this conceptual success, the non-perturbative
structure of QCD in the infrared remains analytically inaccessible. On the
other hand, lower-dimensional quantum field theories have long served as
theoretical laboratories for developing and testing non-perturbative methods.
In particular, many asymptotically free theories in 1+1 dimensions are
integrable, and their exact S-matrices \cite{Za0,KTTW,ZZ} and form factors
\cite{KW,Sm,BFKZ}, obtained through the bootstrap program initiated in the
1970s, enable fully analytic control over scattering amplitudes and
correlation functions.

Previous analysis by Balog and Weisz \cite{Balog1,Balog2} suggested that
integrable 1+1-dimensional models may already capture certain universal
features of deep-inelastic scattering at small Bjorken variable $x$. Their
observations indicated that these theories can exhibit scaling patterns
reminiscent of high-energy behavior in four-dimensional gauge theories,
thereby motivating a closer examination of structure functions in integrable
settings. This provides a broader physical context for studying deep
inelastic-type quantities in exactly solvable models and for exploring to what
extent universal phenomena can arise independently of dimensionality.

In this article, we investigate structure functions in several exactly
integrable quantum field theories in 1+1 dimensions, taking advantage of the
availability of exact form factors. Form factors for a number of relevant
models are known in closed form \cite{KW,Sm,BFKZ}, allowing a detailed
analysis of their responses to external currents. We focus on the small-$x$
regime and find that, in this limit, the structure functions factories into an
$x$-dependent part and a $q^{2}$-dependent part,%
\[
W(x,q^{2})\overset{x\rightarrow0}{\rightarrow}f(x)g(q^{2})
\]
a behavior consistent with the qualitative expectations raised in earlier work
\cite{Balog1,Balog2}. Our goal is to compute these functions explicitly for
several representative integrable models and to characterize their universal
and model-dependent features.

Here we consider the $O(N)$-$\sigma$ model, the $SU(N)$ chiral Gross-Neveu
model, the sine--Gordon model, as well as the sinh-Gordon and the $Z(N)$ Ising
model in the scaling limit. These models have different symmetry groups: the
first two models have $O(N)$ and $SU(N)$ symmetry, while the last three have
$SL(2)_{q}$ quantum-group, $U(1)$ and $Z(N)$ symmetry. We should also mention
that the first two models are asymptotically free.

The exact S-matrices of these two groups of models are solutions of the
Yang-Baxter equation, which are rational and trigonometric, respectively. This
last point is essential for the high-energy asymptotic behavior of the
corresponding S-matrices, which is important in the Balog-Weisz approach. More
precisely, in this article we calculate the structure functions of these two
groups of models using exact form factors. The form factor equations imply
that the structure function factories for small value of $x$ as a product of
two functions, depending separately on $x$ and $q^{2}$.

For the first group of models, which are asymptotically free, we obtain a
universal function of $x$ (at small $x$)%
\[
f(x)\propto\frac{1}{x\ln^{2}x}%
\]
This function is the same for the $O(N)$-$\sigma$ model, as well as for the
$SU(N)$ Gross-Neveu model. Such universal behavior was conjectured by Balog
and Weisz \cite{Balog1,Balog2} as the small Bjorken $x$ behavior for 1+1
dimensional exactly integrable and asymptotically free QFTs. In this article,
we confirm that indeed, the 1+1 dimensional asymptotically free QFT models
exhibit this universal behavior, at least those models we consider. Balog and
Weisz further conjectured that "structurally similar universal formulae may
also hold for the small $x$ behavior of QCD in 4-dimensions". This should be
verified in future investigations in 4 dimensions and also checked
experimentally in the strong-coupling regime.

For the second group of models, we observe a power-law behavior at small $x$.
For models without backward scattering we find the simple behavior
proportional to $f(x)\propto x$. For the sine-Gordon model with $\beta
^{2}>4\pi$ (where there are only solitons and no bound states) we get
\[
f(x)\propto x^{-\lambda(\beta)}%
\]
where $\lambda(\beta)$ is a simple function of the coupling constant $\beta$.
This power-law behavior in $x$ is quite interesting because, as observed in
\cite{BKS}, in some energy interval $q^{2}>20~GeV^{2}$, the structure function
of experimental data for lepton (electron, positron)-proton scattering (HERA
data) also exhibits power-law behavior. This behavior is similar to the
sine-Gordon model when $\beta^{2}\approx6.\,23\pi$, which is a remarkable
observation. In \cite{BKS}, we considered old theoretical ideas and
experimental data related to Regge-calculus amplitudes and DIS amplitudes at
high energy. There is a factorization into two 2-dimensional parts: a
longitudinal and a transversal part. This simplifies their structure, is
essential for our work, and might be crucial for future investigations.
Collecting all these facts, we can conjecture that the longitudinal sector of
QCD at small $x$ might be related to the sine-Gordon model, which possesses
the same power-law behavior mentioned above.

\section{Structure functions}

Following J. Balog and P. Weisz \cite{Balog1,Balog2}, we investigate the
Fourier transform of the 1-particle expectation value for the commutator of
two local operators $\mathcal{O}(x)$ and $\mathcal{O}^{\prime}(x)$%
\begin{equation}
W(p,q)=\int d^{2}x\,e^{iqx}\,\left\langle \theta|\left[  \mathcal{O}%
(x),\mathcal{O}^{\prime}(0)\right]  |\theta\right\rangle \,. \label{W0}%
\end{equation}

\subsection{DIS kinematic variables}

We use the deep inelastic scattering (DIS) variables $\kappa$ and the
Bjorken variable $x$ given by%
\[
x=-\frac{q^{2}}{2pq},~q^{2}=-4\kappa^{2}m^{2}.
\]
If the DIS variables satisfy
\[
q^{2}<0,~\left(  p+q\right)  ^{2}\geq m^{2},~0\leq x\leq1
\]
then the second term in (\ref{W0}) vanishes $\left[  \mathcal{O}%
(x),\mathcal{O}^{\prime}(0)\right]  \rightarrow\mathcal{O}(x)\mathcal{O}%
^{\prime}(0)$ and%
\begin{align}
W(p,q) &  =\int d^{2}xe^{iqx}\,\left\langle \theta|\mathcal{O}(x)\mathcal{O}%
^{\prime}(0)|\theta\right\rangle =\sum_{r}W_{r}(p,q)\nonumber\\
W_{r}(p,q) &  =\int_{\underline{\theta}}J_{r}(\underline{\theta},\theta
)~(2\pi)^{2}\delta^{(2)}\left(  p+q-R\right)  \nonumber\\
J_{r}(\underline{\theta},\theta) &  =\left\langle \theta|\mathcal{O}%
(0)|\underline{\theta}\rangle_{\underline{\alpha}}\,^{\underline{\alpha}%
}\langle\underline{\theta}|\mathcal{O}^{\prime}(0)|\theta\right\rangle
\label{WW}%
\end{align}
with $\int_{\underline{\theta}}=\frac{1}{\left(  4\pi\right)  ^{r}}%
\int_{-\infty}^{\infty}d\theta_{1}\int_{-\infty}^{\theta_{1}}d\theta_{2}%
\dots\int_{-\infty}^{\theta_{r-1}}d\theta_{r}$ and an r-particle state
$|\underline{\theta}\rangle_{\underline{\alpha}}$ with momentum $R=\sum
_{i=1}^{r}p_{i},~p_{i}=m\binom{\cosh\theta_{i}}{\sinh\theta_{i}}~$and
$\underline{\alpha}=\alpha_{1}\dots\alpha_{r},$ where $\alpha_{i}$ denote the
type of the particle. The expression $|\underline{\theta}\rangle
_{\underline{\alpha}}\,^{\underline{\alpha}}\langle\underline{\theta}|$ is
always to be understood as $%
{\textstyle\sum\nolimits_{\underline{\alpha}}}
|\underline{\theta}\rangle_{\underline{\alpha}}\,^{\underline{\alpha}}%
\langle\underline{\theta}|$. The 1-particle state $|\theta\rangle$ has
momentum $p=m\binom{\cosh\theta}{\sinh\theta}$.

\paragraph{New variables:}

For details see \cite{Balog2} and appendix \ref{a0}. We apply the
transformation $\underline{\theta}\rightarrow\underline{u},\Lambda$, defined
by%
\begin{align}
u_{i}  &  =\theta_{i}-\theta_{i+1},~1\leq i\leq r-1\label{nv}\\
\mu(\underline{u})e^{\pm\Lambda}  &  =e^{\pm\theta_{1}}+\dots+e^{\pm\theta
_{r}}\nonumber
\end{align}
with
\begin{equation}
\mu^{2}(\underline{u})=R^{2}/m^{2}=r+2%
{\textstyle\sum\nolimits_{1\leq i<j\leq r}}
\cosh\left(  u_{i}+\dots+u_{j-1}\right)  . \label{mu}%
\end{equation}
We perform the $\Lambda$ and $u_{r-1}$ integrations and obtain
\begin{equation}
W_{r}(p,q)=\frac{1}{(4\pi)^{r-2}}\frac{1}{4m^{2}}\int
d\underline{u}_{r-1}~\frac{1}{\sum_{i=1}^{r-1}\sinh\left(  u_{i}%
+\dots+u_{r-1}^{(0)}\right)  }J_{r}\left(  \underline{\theta}(\underline
{u}^{(0)},\Lambda_{0})\right)  \label{Wr}%
\end{equation}
where $\int d\underline{u}_{r-1}=\int_{0}^{\infty}du_{1}\dots\int_{0}^{\infty
}du_{r-2}$. Here we have defined $\Lambda_{0}=\frac{1}{2}\ln\frac{\left(
p+q\right)  ^{+}}{\left(  p+q\right)  ^{-}}$ and $\underline{u}^{(0)}%
=u_{1},\dots,u_{r-2},u_{r-1}^{(0)}$ where $u_{r-1}^{(0)}$ is a solution of eq.
(\ref{u0}).

\paragraph{$\Lambda_{0}$ and $u_{r-1}^{(0)}$ in terms of $x$ and $\kappa$ for
small $x$:}

We take $p=\binom{m}{0}$, then we have using Lemma \ref{lu-1} of appendix
\ref{a0}%
\begin{align}
\theta_{i}(\underline{u}^{(0)},\Lambda_{0})  &  =u_{i}+\dots+u_{r-2}%
+W,~i=1,\dots,r-1\label{theta0}\\
\theta_{r}\left(  \underline{u}^{(0)},\Lambda_{0}\right)   &  =\epsilon
\nonumber
\end{align}
where for small $x$%
\begin{align}
W  &  =-\ln x+O(1)\label{W}\\
\epsilon &  =x\left(  1+\frac{1}{4\kappa^{2}}\mu_{r-1}^{2}\right)  +O(x^{2}).
\label{eps}%
\end{align}
Here $\mu_{r-1}^{2}$ is given by (\ref{mu}) for $r\rightarrow r-1$. Finally we
have for small $x$
\begin{equation}
W_{r}(p,q)=\frac{1}{\left(  4\pi\right)  ^{r-2}}\frac{x}{4\kappa^{2}}\frac
{1}{2m^{2}}\int d\underline{u}_{r-1}~J_{r}\left(  \underline{\theta
}(\underline{u}^{(0)},\Lambda_{0}),0\right)  +O(x^{2}). \label{Wrx}%
\end{equation}

\section{Factorization}

Form factors are vector functions given by matrix elements of local operators
$\mathcal{O}(x)$%
\[
F_{\underline{\alpha}}^{\mathcal{O}}(\underline{\theta})=\langle
0|\mathcal{O}(0)|\underline{\theta}\rangle_{\underline{\alpha}}%
\]
with $\underline{\theta}=\theta_{1},\dots,\theta_{n}$ and $\underline{\alpha
}=\alpha_{1},\dots,\alpha_{n}$. These functions satisfy the form factor
equations\footnote{For the chiral $SU(N)$ Gross-Neveu model these equations
are modified because of the anyonic statistics of the particles (see appendix
\ref{apSUN}).} (see eg. \cite{KW,Sm,BFKZ})
\begin{align}
F_{\dots\alpha_{i}\alpha_{j}\dots}^{\mathcal{O}}(\dots,\theta_{i},\theta
_{j},\dots)  &  =F_{\dots\alpha_{j}^{\prime}\alpha_{i}^{\prime}\dots
}^{\mathcal{O}}(\dots,\theta_{j},\theta_{i},\dots)\,S_{\alpha_{i}\alpha_{j}%
}^{\alpha_{j}^{\prime}\alpha_{i}^{\prime}}(\theta_{ij})\label{fi}\\
^{\alpha}\langle\theta|\mathcal{O}|\underline{\theta}\rangle_{\underline
{\alpha}}  &  =\langle0|\mathcal{O}|\underline{\theta},\theta-i\pi
\rangle_{\underline{\alpha}\alpha^{\prime}}\mathbf{C}^{\alpha^{\prime}\alpha
}=F_{\underline{\alpha}\alpha^{\prime}}^{\mathcal{O}}(\underline{\theta
},\theta-i\pi)\,\mathbf{C}^{\alpha^{\prime}\alpha}\label{fii}\\
\operatorname*{Res}_{\theta_{12}=i\pi}F_{\underline{\alpha}}^{\mathcal{O}%
}(\underline{\theta})  &  =2i\,\mathbf{C}_{\alpha_{1}^{\prime}\alpha
_{2}^{\prime}}\,F_{\underline{\hat{\alpha}}^{\prime}}^{\mathcal{O}}%
(\underline{\hat{\theta}})\left(  \mathbf{1}-S_{2n}\dots S_{23}\right)
_{\underline{\alpha}}^{\underline{\alpha^{\prime}}}\label{fiii}\\
F_{\underline{\alpha}}^{\mathcal{O}}(\underline{\theta})  &  =e^{s\mu
}\,F_{\underline{\alpha}}^{\mathcal{O}}(\underline{\theta}+\mu) \label{fv}%
\end{align}
where $S(\theta)$ is the S-matrix, $\mathbf{C}$ the charge conjugation matrix,
$\underline{\hat{\theta}}=\theta_{3},\dots,\theta_{n}$, $\underline
{\hat{\alpha}}=\alpha_{3},\dots,\alpha_{n}$ and $s$ the spin of $\mathcal{O}$.

These relations imply%
\begin{align*}
^{\alpha}\langle\theta|\mathcal{O}(0)|\underline{\theta}\rangle_{\underline
{\alpha}}  & =F_{\underline{\hat{\alpha}}\alpha_{r}\alpha^{\prime}%
}^{\mathcal{O}}(\underline{\theta}_{r-1},\theta_{r},-i\pi)\mathbf{C}%
^{\alpha^{\prime}s}\\
& \overset{\theta_{r}\rightarrow0}{\rightarrow}\frac{2i}{\theta_{r}%
}F_{\underline{\hat{\alpha}}^{\prime}}^{\mathcal{O}}(\underline{\theta}%
_{r-1})\left(  \mathbf{1}-S(\theta_{1r})\dots S(\theta_{r-1r})\right)
_{\underline{\hat{\alpha}}\alpha_{r}}^{\alpha\underline{\hat{\alpha}}^{\prime
}}%
\end{align*}
Further we need the asymptotic behavior of the monodromy matrices%
\begin{gather}
T_{\underline{\delta}\gamma}^{\alpha\underline{\beta}}(\underline{\theta
}+W,\theta)=\left(  S(\theta_{1}+W-\theta)\dots S(\theta_{n}+W-\theta)\right)
_{\underline{\delta}\gamma}^{\alpha\underline{\beta}}\nonumber\\
~~~~~~~~~~~~\overset{W\rightarrow\infty}{\rightarrow}\left(
\mathbf{1\mathbf{+}}X\mathbf{(}W-\theta\mathbf{)M}(\underline{\theta})\right)
_{\underline{\delta}\gamma}^{\alpha\underline{\beta}}\label{ta}\\
T_{\alpha\underline{\beta}}^{\underline{\delta}\gamma}(\theta,\underline
{\theta}+W)=\left(  S(\theta-\theta_{1}-W)\dots S(\theta-\theta_{n}-W)\right)
_{\alpha\underline{\beta}}^{\underline{\delta}\gamma}\nonumber\\
~~~~~~~~~~~~\overset{W\rightarrow\infty}{\rightarrow}\left(  \mathbf{1+}%
X^{\ast}(W-\theta)\mathbf{M}(\underline{\theta})\right)  _{\alpha
\underline{\beta}}^{\underline{\delta}\gamma} \label{tab}%
\end{gather}
with the scalar function $X\mathbf{(}W\mathbf{)}$ and the Matrix
\begin{equation}
\text{ }\mathbf{M}(\underline{\theta})=\sum_{i=1}^{m}\left(  1\dots
M_{i}(\theta_{i})\dots1\right)  \label{M}%
\end{equation}
see Appendix.

\begin{lemma}
\label{las}We set $p=\binom{m}{0}$ or $\theta=0$, write $\underline{\theta
}=\underline{\hat{\theta}}+W,\theta_{r}$ and $\underline{\alpha}=\alpha
_{1},\dots,\alpha_{r}=\underline{\hat{\alpha}},\alpha_{r}$. For Lorentz scalar
operators $\mathcal{O}$ and $\mathcal{O}^{\prime}$ the form factor equations
(\ref{fi})-(\ref{fv}) and (\ref{ta}) imply for $x\rightarrow0$ (which means
$\theta_{r}\rightarrow0$ and $W\rightarrow\infty$ see (\ref{theta0}),
(\ref{W}) and (\ref{eps}))%
\begin{align}
&  ^{\alpha}\langle\theta|\mathcal{O}(0)|\underline{\theta}\rangle
_{\underline{\alpha}}\overset{x\rightarrow0}{\rightarrow}-\frac{2i}{\theta
_{r}}X(W)\langle0|\mathcal{O}(0)|\underline{\hat{\theta}}\rangle
_{\underline{\hat{\alpha}}^{\prime}}\mathbf{M}(\underline{\hat{\theta}%
})_{\underline{\hat{\alpha}}\alpha_{r}}^{\alpha\underline{\hat{\alpha}%
}^{\prime}}\label{tW1}\\
&  ^{\underline{\alpha}}\langle\underline{\theta}|\mathcal{O}^{\prime
}(0)|\theta\rangle_{\beta}\overset{x\rightarrow0}{\rightarrow}\frac{2i}%
{\theta_{r}}X^{\ast}(W)\mathbf{M}(\underline{\hat{\theta}})_{\beta
\underline{\hat{\alpha}}^{\prime}}^{\underline{\hat{\alpha}}\alpha_{r}%
}~^{\underline{\hat{\alpha}}^{\prime}}\langle\underline{\hat{\theta}%
}|\mathcal{O}^{\prime}(0)|0\rangle\label{tW2}%
\end{align}
and finally%
\begin{multline}
^{\alpha}\langle\theta|\mathcal{O}(0)|\underline{\theta}\rangle_{\underline
{\alpha}}~^{\underline{\alpha}}\langle\underline{\theta}|\mathcal{O}^{\prime
}(0)|\theta\rangle_{\beta}~\overset{x\rightarrow0}{\rightarrow}\left(
\frac{2}{\theta_{r}}\right)  ^{2}X(W)X^{\ast}(W)\label{XX}\\
\times\langle0|\mathcal{O}(0)|\underline{\hat{\theta}}\rangle_{\underline
{\hat{\alpha}}^{\prime\prime}}\mathbf{M}(\underline{\hat{\theta}}%
)_{\underline{\hat{\alpha}}\alpha_{r}}^{\alpha\underline{\hat{\alpha}}%
^{\prime\prime}}~\mathbf{M}(\underline{\hat{\theta}})_{\beta\underline
{\hat{\alpha}}^{\prime}}^{\underline{\hat{\alpha}}\alpha_{r}}~^{\underline
{\hat{\alpha}}^{\prime}}\langle\underline{\hat{\theta}}|\mathcal{O}^{\prime
}(0)|0\rangle.
\end{multline}

\end{lemma}

\begin{proof}
See appendices \ref{apON}, \ref{apSUN} and \ref{apSG}. Note that eq.
(\ref{XX}) also holds, if $\mathcal{O}$ and $\mathcal{O}^{\prime}$ are not
Lorentz scalar operators, but have the same spin, because the contributions of
form factor equation (\ref{fv}) would cancel.
\end{proof}

\begin{theorem}
\label{th}Let $\mathcal{O=O}^{A},~\mathcal{O}^{\prime}=\mathcal{O}_{B}$ local
operators, where $A,B$ denote an iso-vector, an antisymmetric tensor or
adjoint representation symmetry. The structure function (\ref{W0}) factorizes
for small $x$%
\begin{equation}
\fbox{$\rule[-0.2in]{0in}{0.5in}\displaystyle~W_{\beta,B}^{\alpha
,A}(p,q)\overset{x\rightarrow0}{\rightarrow}f(x)g_{\beta,B}^{\alpha,A}(q^{2})$
~} \label{W1}%
\end{equation}
where%
\begin{align}
f(x) &  =\frac{1}{x}\left\vert X(-\ln x)\right\vert ^{2}\label{f}\\
g_{\beta,B}^{\alpha,A}(q^{2}) &  =\frac{2}{-q^{2}}\sum_{r}\frac{1}{\left(
4\pi\right)  ^{r-2}}\int d\underline{u}_{r-1}w_{\beta,B}^{\alpha,A}%
(\underline{u}_{r-1})\label{g}\\
w_{\beta,B}^{\alpha,A}(\underline{u}_{r-1}) &  =\frac{\langle0|\mathcal{O}%
^{A}(0)|\underline{\hat{\theta}}\rangle_{\underline{\hat{\alpha}}%
^{\prime\prime}}\mathbf{M}(\underline{\hat{\theta}})_{\underline{\hat{\alpha}%
}\alpha_{r}}^{\alpha\underline{\hat{\alpha}}^{\prime\prime}}~\mathbf{M}%
(\underline{\hat{\theta}})_{\beta\underline{\hat{\alpha}}^{\prime}%
}^{\underline{\hat{\alpha}}\alpha_{r}}~^{\underline{\hat{\alpha}}^{\prime}%
}\langle\underline{\hat{\theta}}|\mathcal{O}_{B}(0)|0\rangle}{\left(
1+\mu_{r-1}^{2}/(4\kappa^{2})\right)  ^{2}}\label{w}%
\end{align}
with $\underline{u}_{r-1}=u_{1},\dots,u_{r-2}$.
\end{theorem}

\begin{proof}
The claim follows from (\ref{Wrx}) and (\ref{XX}) using (\ref{theta0}),
(\ref{W}) and (\ref{eps}).
\end{proof}

\begin{remark}
For the $SU(N)$ chiral Gross-Neveu model the form factor equations and
(\ref{ta}) are modified because of the anyonic statistics of the particles.
However, for charge zero operators the results of Theorem (\ref{th}) and eqs.
(\ref{f}), (\ref{g}) and (\ref{w}) hold unchanged, see subsection \ref{sSUN}
and appendix \ref{apSUN}. This is also true for the sine-Gordon model for
$\nu>1$, the sinh-Gordon and the $Z(N)$-Ising model in the scaling limit, see
subsections \ref{sSG}, \ref{ssinhG}, \ref{sZN} and appendices \ref{apSG},
\ref{apsinhG}, \ref{apZN}.
\end{remark}

\section{Results for various models}

The functions $X(W)$ and $f(x)$ depend on the 2-particle S-matrix as solution
of Yang-Baxter equations. We distinguish the following cases:

\noindent{\bf Rational solutions:}
The $O(N)$-$\sigma$ and the $SU(N)$ chiral Gross-Neveu model show the typical
asymptotic behavior (\ref{ta}) and (\ref{tab}) with (see appendices
\ref{apON},\ref{apSUN})
\begin{equation}
X(W)=\mathbf{-}i\pi\nu\frac{1}{W} \label{XONSUN}%
\end{equation}
and the function (\ref{f}) is%
\begin{equation}
f(x)=\pi^{2}\nu^{2}\frac{1}{x\ln^{2}x} \label{fONSUN}%
\end{equation}
where $\nu=2/(N-2)$ for $O(N)$ and $\nu=2/N$ for $SU(N)$. The matrix $M$ in
(\ref{M}) is independent of $\theta$ and is the Lie-algebra generator in the
vector representation and $\mathbf{M}$ in higher representations.

For $O(N)$ it is%
\begin{equation}
M_{\alpha\beta}^{\delta\gamma}=\left(  P-K\right)  _{\alpha\beta}%
^{\delta\gamma}=\delta_{\alpha}^{\gamma}\delta_{\beta}^{\delta}-\mathbf{C}%
^{\delta\gamma}\mathbf{C}_{\alpha\beta}\,,~~\nu=2/(N-2) \label{MON}%
\end{equation}
where the charge conjugation matrix $\mathbf{C}$ for the real basis is
$\mathbf{C}_{\alpha\beta}=\delta_{\alpha\beta}$ (see appendix \ref{apON}).

For the $SU(N)$ case the form factor equations are more complicated, because
the statistics of the particles is anyonic. However, if $\underline{\beta}$ is
a chargeless state\footnote{This means there are the same number of particles
and anti-particles.}, the relations (\ref{ta}) and (\ref{tab}) hold with
(\ref{XONSUN}) and
\begin{equation}
M_{\alpha\beta}^{\delta\gamma}=P_{\alpha\beta}^{\delta\gamma}=\delta_{\alpha
}^{\gamma}\delta_{\beta}^{\delta}~,~~M_{\bar{\alpha}\beta}^{\delta\bar{\gamma
}}=-\mathbf{C}^{\delta\bar{\gamma}}\mathbf{C}_{\bar{\alpha}\beta}\,,~\nu=2/N
\label{MSUN}%
\end{equation}
(see appendix \ref{apSUN}).

\noindent{\bf Trigonometric solution with backward scattering:}
The Sine-Gordon model is equivalent to the massive Thirring model \cite{Co},
which is symmetric with respect to the "quantum group $SL_{q}(2)$" with
\begin{equation}
q=e^{-i\pi\left(  1-1/\nu\right)  }~,~~\nu=\frac{\beta^{2}}{8\pi-\beta^{2}}.
\label{nuSG}%
\end{equation}
We consider in the following the case $\nu>1$, where there are only solitons
and anti-solitons and no bound states ($\alpha,\beta,\gamma,\delta\in
\{s,\bar{s}\}$). The typical asymptotic behavior is given by the formulae
(\ref{ta}) and (\ref{tab}), if $\underline{\beta}$ is a chargeless state (see
appendix \ref{apSG}), with%
\begin{align}
X(W)  &  =2i\left(  \sin\pi/\nu\right)  e^{-W/\nu}\label{XSG}\\
M_{s\bar{s}}^{s\bar{s}}(\theta)  &  =M_{\bar{s}s}^{\bar{s}s}(\theta
)=e^{-\theta/\nu}\,,~M_{\alpha\beta}^{\delta\gamma}=0~\text{else.} \label{MSG}%
\end{align}
The function (\ref{f}) shows power behavior%
\begin{equation}
f(x)=\left(  4\sin^{2}\pi/\nu\right)  \,x^{\frac{2}{\nu}-1},~\nu>1.
\label{fSG}%
\end{equation}

\begin{remark}
For $\beta^{2}\rightarrow8\pi$ or $\nu\rightarrow\infty$ the symmetry is given
by $SL_{-1}(2)$. The model will be similar to the $SU(2)$ chiral Gross-Neveu
model, asymptotically free and the small $x$ behavior is that of
(\ref{fONSUN}) $f(x)\propto x^{-1}\ln^{-2}x$.
\end{remark}

\begin{remark}
The result (\ref{fSG}) may be generalized: For the $SL_{q}(N)$ symmetric
S-matrix, as $q$-deformation of the $SU(N)$ S-matrix given by eqs. (\ref{sun})
and (\ref{aSUN}), the function $f(x)$ is%
\begin{equation}
f(x)=4\left(  \sin^{2}\pi\mu\right)  x^{N\mu-1},~q=e^{-i\pi\mu}\,.
\label{fSUNq}%
\end{equation}
\end{remark}

\noindent{\bf Trigonometric solution without backward scattering:}
For these models the S-matrix is a number and the typical behavior is given
by
\[
X(W)\propto e^{-W},~~f(x)\propto x,~~M(\underline{\theta})=\sum_{i=1}%
^{n}e^{-\theta_{i}}%
\]
where $M$ is a number here. In particular:

\subparagraph{The sinh-Gordon model S-matrix}

is obtained from the sine-Gordon breather S-matrix for imaginary couplings
$\beta$. There is only one type of bosonic particles and the model is $U(1)$ symmetric.

The functions $X$ and $f$ are (see appendix \ref{apsinhG})%

\[
X(W)=\left(  4i\sin\pi\nu\right)  e^{-W},~~f(x)=\left(  16\sin^{2}\pi
\nu\right)  \,x
\]

\subparagraph{The $Z(N)$ model S-matrix}

is \cite{KS,K3,BFK,BK04}%
\[
S(\theta)=\frac{\sinh\frac{1}{2}(\theta+\frac{2\pi i}{N})}{\sinh\frac{1}%
{2}(\theta-\frac{2\pi i}{N})}%
\]

and the functions $X$ and $f$ are (see appendix \ref{apZN})%
\[
X(W)=\left(  2i\sin2\pi/N\right)  e^{-W},~~f(x)=\left(  4\sin^{2}%
2\pi/N\right)  \,x
\]

\section{The function $g(q^{2})$}

In the following, we compute the function $g_{\beta,B}^{\alpha,A}(q^{2})$ of
(\ref{g}) for the currents of the models: $O(N)~\sigma$-, $SU(N)$ chiral
Gross-Neveu, Sine-Gordon and $Z(N)$-Ising in the scaling limit, in
three-intermediate-particle approximation.

We investigate the structure function of the current $J_{\mu}(x)$%
\[
W_{\mu\nu}(p,q)=\int d^{2}xe^{iqx}\int_{\underline{\theta}}\left\langle
\theta|\left[  J_{\mu}(x),J_{\nu}(0)\right]  |\theta\right\rangle
\]
(for simplicity we suppress here the isospin indices). The conservation law
$\partial^{\mu}J_{\mu}(x)=0$ implies the existence of a pseudo-potential
$\varphi(x)$ with%
\[
J_{\mu}(x)=\epsilon_{\mu\nu}\partial^{\nu}\varphi(x).
\]
Therefore we consider%
\begin{align}
W_{\mu}^{\mu}(p,q)  &  =\int d^{2}xe^{iqx}\int_{\underline{\theta}%
}~\left\langle \theta|J_{\mu}(x)J^{\mu}(0)|\theta\right\rangle =\sum_{r}%
W_{r}(p,q)\nonumber\\
W_{r}(p,q)  &  =-q^{2}\int_{\underline{\theta}}~\langle\theta|\varphi
(0)|\underline{\theta}\rangle_{\underline{\alpha}}~^{\underline{\alpha}%
}\langle\underline{\theta}|\varphi(0)|\theta\rangle~(2\pi)^{2}\delta
^{(2)}\left(  p+q-R\right)  \label{Wgen}%
\end{align}
where the current is written in terms of the pseudo-potential and where
$g^{\mu\nu}\epsilon_{\mu\mu^{\prime}}\epsilon_{\nu\nu^{\prime}}q^{\mu^{\prime}}%
q^{\nu^{\prime}}=-q^{2}$ was used. For $r=3$ the function of (\ref{g}) reads
as
\begin{equation}
g(q^{2})=\frac{1}{2\pi}\int_{0}^{\infty}du\,w(u)\,. \label{g3}%
\end{equation}
with
\begin{equation}
w_{\beta}^{\alpha}(u)=w(u)\delta_{\beta}^{\alpha}=\frac{\langle0|\varphi
(0)|\underline{\hat{\theta}}\rangle_{\underline{\hat{\alpha}}^{\prime\prime}%
}\mathbf{M}(\underline{\hat{\theta}})_{\underline{\hat{\alpha}}\alpha_{r}%
}^{\alpha\underline{\hat{\alpha}}^{\prime\prime}}~\mathbf{M}(\underline
{\hat{\theta}})_{\beta\underline{\hat{\alpha}}^{\prime}}^{\underline
{\hat{\alpha}}\alpha_{r}}~^{\underline{\hat{\alpha}}^{\prime}}\langle
\underline{\hat{\theta}}|\varphi(0)|0\rangle}{\left(  1+\left(  2+2\cosh
u\right)  /(-q^{2}/m^{2})\right)  ^{2}} \label{w3}%
\end{equation}
with $\underline{\hat{\theta}}=u,0$ (see (\ref{theta0})), the matrix
$\mathbf{M}(\underline{\hat{\theta}})$ depends on the model.

\subsection{The $O(N)$ $\sigma$-model}

\label{sON}

\paragraph{$O(N)$ current and pseudo-potential:}

The pseudo-potential $J^{\gamma\delta}(x)$ is an antisymmetric tensor and is
related to the current by%
\[
J_{\mu}^{\gamma\delta}(x)=\varphi^{\gamma}\partial_{\mu}\varphi^{\delta
}-\varphi^{\delta}\partial_{\mu}\varphi^{\gamma}=\epsilon_{\mu\nu}%
\partial^{\nu}J^{\gamma\delta}(x).
\]
In particular, we have for two particles \cite{KW,BFK7}%
\begin{equation}
\langle0|J^{\gamma\delta}\left(  0\right)  |\underline{\theta}\rangle
_{\alpha_{1}\alpha_{2}}=\left(  \delta_{\alpha_{1}}^{\gamma}\delta_{\alpha
_{2}}^{\delta}-\delta_{\alpha_{1}}^{\delta}\delta_{\alpha_{2}}^{\gamma
}\right)  \tanh\tfrac{1}{2}\theta_{12}~F_{-}(\theta_{12}) \label{J2}%
\end{equation}
where $F_{-}(\theta)$ is the minimal form factor function (\ref{FON}) . For
convenience we introduce%
\[
J_{\rho\sigma}=\mathbf{C}_{\rho\delta}\mathbf{C}_{\sigma\gamma}J^{\gamma
\delta}%
\]
where the charge conjugation matrix is $\mathbf{C}_{\alpha\beta}%
=\delta_{\alpha\beta}$ \cite{BFK7}. For details see appendix \ref{apON}.

\subparagraph{Three-intermediate-particle $(r=3)$ approximation:}

We consider the structure function of theorem \ref{th} $W_{\beta,B}^{\alpha
,A}(p,q)=W_{\beta,\rho\sigma}^{\alpha,\gamma\delta}(p,q)$ and the special
component where $~\alpha=\beta=1,~\gamma=\rho=2,~\delta=\sigma=3$. Then using
(\ref{J2}) the function $w(u)$ in (\ref{w3}) is%
\begin{equation}
w(u)=\frac{4\left(  \tanh^{2}\tfrac{1}{2}u\right)  }{\left(  1+\left(
2+2\cosh u\right)  /(-q^{2}/m^{2})\right)  ^{2}}F_{-}(u)F_{-}(-u) \label{wON}%
\end{equation}
because only $\mathbf{M}_{13,2}^{1,23}~\mathbf{M}_{1,23}^{13,2}+\mathbf{M}%
_{21,3}^{1,32}~\mathbf{M}_{1,23}^{21,3}+\mathbf{M}_{12,3}^{1,32}%
~\mathbf{M}_{1,32}^{12,3}+\mathbf{M}_{31,2}^{1,32}~\mathbf{M}_{1,32}^{31,2}=4$
contributes to $\mathbf{M}_{\underline{\hat{\alpha}}\alpha_{r}}^{1\underline
{\hat{\alpha}}^{\prime\prime}}~\mathbf{M}_{1\underline{\hat{\alpha}}^{\prime}%
}^{\underline{\hat{\alpha}}\alpha_{r}}$ (see (\ref{M}) and (\ref{MON})). The
function $g(q^{2})$ for the $O(N)$ $\sigma$-model is given by (\ref{g3}),
(\ref{wON}) and is plotted in Figure \ref{fon}.
\begin{figure}
[h]
\begin{center}
\includegraphics[
height=2.2733in,
width=3.6624in
]%
{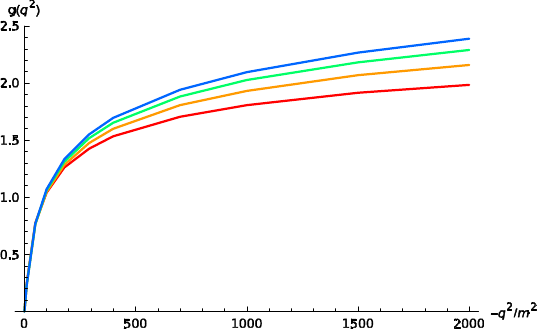}%
\\
\caption{Plots of $g(q^{2})$ versus $-q^{2}/m^{2}$ for $O(N)$, where $N=$ 3
(red), 4 (yellow), 5 (green) and 6 (blue)}%
\label{fon}%
\end{center}
\end{figure}

\subsection{The $SU(N)$ chiral Gross-Neveu model}

\label{sSUN}

\paragraph{The $SU(N)$ Noether current:}

The pseudo-potential $J^{\alpha\bar{\beta}}(x)$ is related to the $SU(N)$
Noether current by%
\begin{equation}
J_{\mu}^{\alpha\bar{\beta}}=\epsilon_{\mu\nu}\partial^{\nu}J^{\alpha\bar
{\beta}}(x) \label{J}%
\end{equation}
The particle anti-particle state the form factor is \cite{BFK1,BFK3}%
\begin{equation}
\langle\,0\,|\,J^{\gamma\bar{\delta}}|\,\underline{\theta}\rangle_{\alpha
\bar{\beta}}=F_{\alpha\bar{\beta}}^{J^{\gamma\bar{\delta}}}(\underline{\theta
})=\delta_{\alpha}^{\gamma}\delta_{\bar{\beta}}^{\bar{\delta}}\left(
\cosh\tfrac{1}{2}\theta_{12}\right)  ^{-1}\bar{F}\left(  \theta_{12}\right)
\label{J2SUN}%
\end{equation}
where $\bar{F}\left(  \theta\right)  $ is the minimal form factor function
(\ref{FSUN}).For convenience we define%
\[
J_{\rho\bar{\sigma}}=\mathbf{C}_{\rho\bar{\delta}}\mathbf{C}_{\bar{\sigma
}\gamma}J^{\gamma\bar{\delta}}.
\]
For details see appendix \ref{apSUN}.

\subparagraph{Three-intermediate-particle $(r=3)$ approximation:}

We consider the structure function of theorem \ref{th} $W_{\beta,B}^{\alpha
,A}(p,q)=W_{\beta,\rho\bar{\sigma}}^{\alpha,\gamma\bar{\delta}}(p,q)$ and the
special component where $~\alpha=\beta=1,~\gamma=\rho=2,~\bar{\delta}%
=\bar{\sigma}=\bar{3}$ . Then using (\ref{J2SUN}) the function $w(u)$ in
(\ref{w3}) is%
\begin{equation}
w(u)=\frac{2\left(  \cosh\tfrac{1}{2}u\right)  ^{-2}}{\left(  1+\left(
2+2\cosh u\right)  /(-q^{2}/m^{2})\right)  ^{2}}\bar{F}(u)\bar{F}(-u)
\label{wSUN}%
\end{equation}
because only $\mathbf{M}_{1\bar{3},2}^{1,2\bar{3}}~\mathbf{M}_{1,2\bar{3}%
}^{1\bar{3},2}=1$ and $\mathbf{M}_{\bar{3}1,2}^{1,\bar{3}2}~\mathbf{M}%
_{1,\bar{3}2}^{\bar{3}1,2}=1$ contribute to $\mathbf{M}_{\underline{\hat
{\alpha}}\alpha_{r}}^{1\underline{\hat{\alpha}}^{\prime\prime}}~\mathbf{M}%
_{1\underline{\hat{\alpha}}^{\prime}}^{\underline{\hat{\alpha}}\alpha_{r}}$
(see (\ref{M}) and (\ref{MSUN})). The function $g(q^{2})$ for the
$SU(N)$-model is given by (\ref{g3}), (\ref{wSUN}) and is plotted in Figure
\ref{fsu}.
\begin{figure}
[h]
\begin{center}
\includegraphics[
height=2.3593in,
width=3.8296in
]%
{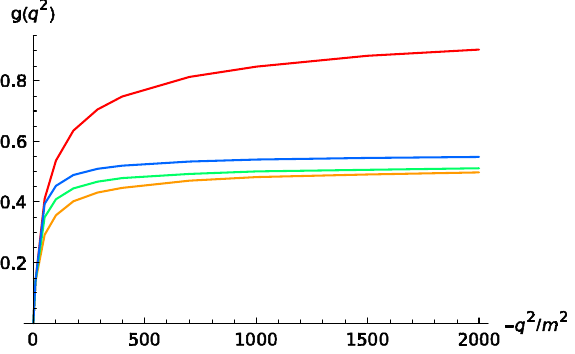}%
\caption{Plots of $g(q^{2})$ versus $-q^{2}/m^{2}$ for $SU(N)$, where $N=$ 2
(red), 3 (yellow), 5 (green) and 8 (blue)}%
\label{fsu}%
\end{center}
\end{figure}

\subsection{Sine-Gordon Model}

\label{sSG}

The classical sine-Gordon model is given by the wave equation
\[
\Box\varphi(t,x)+\frac{\alpha}{\beta}\sin\beta\varphi(t,x)=0.
\]
We consider the case
\[
\nu=\frac{\beta^{2}}{8\pi-\beta^{2}}>1
\]
where there are only solitons and no breathers. The asymptotic behavior of the
monodromy matrices is given by the formulae (\ref{ta}) and (\ref{tab}), if
$\underline{\beta}$ is a chargeless state (see appendix \ref{apSG}), with
$X(W)$ and $M_{\alpha\beta}^{\delta\gamma}(\theta)$ given by (\ref{XSG}) and
(\ref{MSG}).

\paragraph{The current and pseudo-potential:}

The current is related to the sine-Gordon field by \cite{Co}
\begin{equation}
\epsilon^{\mu\nu}\partial_{\nu}\varphi=-\frac{2\pi}{\beta}j^{\mu}. \label{Col}%
\end{equation}
The soliton-antisoliton form factors are (see e.g. \cite{KW,BFKZ})%
\begin{equation}
\langle\,0\,|\,\varphi(0)|\,\underline{\theta}\,\rangle_{s\bar{s}}=\frac{2\pi
}{\beta\cosh\frac{1}{2}\theta_{12}}\frac{\cosh\frac{1}{2}(i\pi-\theta)}%
{\cosh\frac{1}{2}(i\pi-\theta)/\nu}F(\theta_{12}) \label{sg2}%
\end{equation}
where $F(\theta)$ is the minimal form factor function (\ref{FSG}).

\subparagraph{Three-intermediate-particle $(r=3)$ approximation:}

We consider the structure function of theorem \ref{th} $W_{\beta,B}^{\alpha
,A}(p,q)=W_{\beta}^{\alpha}(p,q)$ and the special component where
$~\alpha=\beta=s=$ soliton (as in (\ref{Wgen})). Then using (\ref{sg2}),
(\ref{w}) the function $w_{s}^{s}(u)$ in (\ref{g3}) is for $\nu>1$%
\begin{equation}
w(u)=\frac{1+e^{-2u/\nu}}{\left(  1+\left(  2+2\cosh u\right)  /(-q^{2}%
/m^{2})\right)  ^{2}}\frac{2\tanh^{2}\frac{1}{2}u}{\left(  \cosh\frac{u}{\nu
}+\cos\frac{\pi}{\nu}\right)  }F(u)\,F(-u) \label{wSG}%
\end{equation}
because by (\ref{MSG}) $\mathbf{M}(\underline{\hat{\theta}})_{ss,\bar{s}%
}^{s.s\bar{s}}~\mathbf{M}(\underline{\hat{\theta}})_{s.s\bar{s}}^{ss,\bar{s}%
}+\mathbf{M}(\underline{\hat{\theta}})_{ss,\bar{s}}^{s.\bar{s}s}%
~\mathbf{M}(\underline{\hat{\theta}})_{s.\bar{s}s}^{ss,\bar{s}}=e^{-2\hat
{\theta}_{2}/\nu}+e^{-2\hat{\theta}_{1}/\nu}=1+e^{-2u/\nu}$ (see
(\ref{theta0}) and Lemma \ref{las}). The function $g(q^{2})$ for the
Sine-Gordon model is given by (\ref{g3}), (\ref{wSG}) and is plotted in Figure
\ref{fsg}.%
\begin{figure}
[h]
\begin{center}
\includegraphics[
height=2.3237in,
width=3.7628in
]%
{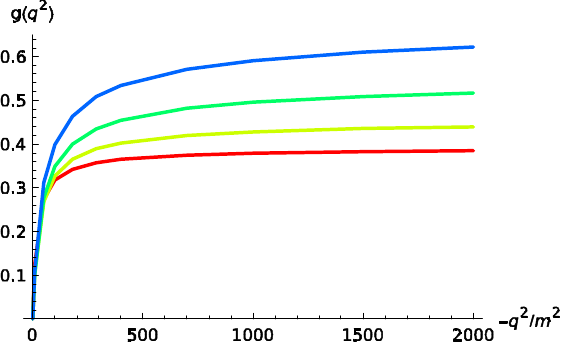}
\caption{Plots of $g(q^{2})$ versus $-q^{2}/m^{2}$ for Sine-Gordon, where
$\nu=$ 1.2 (red), 1.8 (yellow), 3 (green) and 6 (blue)}%
\label{fsg}%
\end{center}
\end{figure}

\subsection{Sinh-Gordon Model}

\label{ssinhG}

\label{1a}
The Sinh-Gordon S-matrix is obtained from the sine-Gordon breather S-matrix
for imaginary couplings $\beta$
\begin{equation}
S(\theta)=\frac{\sinh\theta+i\sin\pi\nu}{\sinh\theta-i\sin\pi\nu}%
\quad\text{with }-1\leq\nu=\frac{\beta^{2}}{8\pi-\beta^{2}}\leq0\,.\label{SSH}%
\end{equation}
Similarly, as above for the sine-Gordon case, we consider the structure
function%
\[
W(p,q)=-q^{2}\int_{\underline{\theta}}~\langle\theta|\varphi(0)|\underline
{\theta}\rangle~\langle\underline{\theta}|\varphi(0)|\theta\rangle~(2\pi
)^{2}\delta^{(2)}\left(  p+q-R\right)  .
\]
The form factors are non-zero only for $r=$ even number of particles.

\subparagraph{Two-intermediate-particle $(r=2)$ approximation:}

The function $g(q^{2})$ of (\ref{g}) is for $-1\leq\nu\leq0$%

\begin{equation}
g(q^{2})=2\frac{1}{\left(  1+1/(-q^{2}/m^{2})\right)  ^{2}}Z^{\varphi}
\label{gsinhg}%
\end{equation}
because $\langle0|\varphi(0)|\hat{\theta}\rangle=\sqrt{Z^{\varphi}}$ is
constant (see \cite{KW}) and by (\ref{Msh}) $\mathbf{M}(\hat{\theta
})~\mathbf{M}(\hat{\theta})=e^{-2\hat{\theta}}=1$ (see appendix \ref{apsinhG}).
The function $g(q^{2})$ is plotted in Figure \ref{fsh}.%
\begin{figure}
[h]
\begin{center}
\includegraphics[
height=1.8396in,
width=2.9776in
]%
{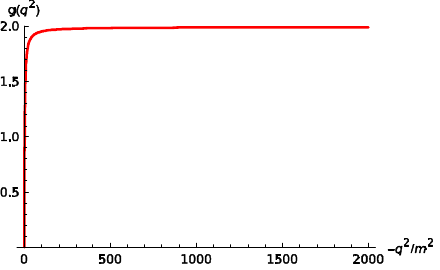}%
\caption{Plot of $g(q^{2})$ versus $-q^{2}/m^{2}$ for Sinh-Gordon, where $\nu=-1/2$}%
\label{fsh}%
\end{center}
\end{figure}

\subsection{Z(N) Model}
\label{sZN}
The $Z(N)$-Ising model in the scaling limit possesses $N-1$ types of
particles: $\alpha=1,\dots,N-1$ of charge $\alpha$ and $\bar{\alpha}=N-\alpha$
is the antiparticle of $\alpha$. The n-particle S-matrix factories in terms
of two-particle ones since the model is integrable.

\paragraph{The Z(N) model S-matrix}

The two-particle S-matrix for the $Z(N)$-Ising model has been proposed by
K\"{o}berle and Swieca \cite{KS}. The S-matrices for the scattering of two
particles of type $1$ and for antiparticle-particle are%
\begin{equation}
S(\theta)=S_{11}(\theta)\,=\frac{\sinh\frac{1}{2}(\theta+i\eta)}{\sinh\frac
{1}{2}(\theta-i\eta)},~\bar{S}(\theta)=S_{\bar{1}1}(\theta)=\frac{\cosh
\frac{1}{2}(\theta-i\eta)}{\cosh\frac{1}{2}(\theta+i\eta)}\label{SZN}%
\end{equation}
with $\eta=2\pi/N$.

\paragraph{The currents $J_{\pm1}^{\mu}(x)$}

are the first ones of the higher currents $J_{L}^{\mu}(x)$ (see \cite{BFK}).
They may be written in terms of the pseudo potentials $\varphi_{\pm1}(x)$ as%
\begin{equation}
J_{\pm1}^{\mu}(x)=\epsilon^{\mu\nu}\partial_{\nu}\varphi_{\pm1}(x) \label{Jpm}%
\end{equation}
In particular, we have the anti-particle particle form factors for $Z(3)$%
\begin{equation}
\langle0|\varphi_{\pm1}\left(  0\right)  |\theta_{1},\theta_{2}\rangle
_{\bar{1}1}=\mp\frac{e^{\pm\frac{1}{2}\left(  \theta_{1}+\theta_{2}\right)  }%
}{\cosh\frac{1}{2}\theta_{12}}\bar{F}(\theta_{12})~ \label{Fphipm}%
\end{equation}
with the minimal form factor function $\bar{F}(\theta)$ of (\ref{FZN}). The
eigenvalues of the "charges" $Q_{\pm1}=\int dxJ_{\pm1}^{0}(x)$ are
$e^{\pm\theta}$.

\subparagraph{Three-intermediate-particle $(r=3)$ approximation:}

\label{1}

We consider the structure function of the higher currents $J_{\pm 1}^{\mu}(x)$
similar as in (\ref{Wgen}) and the function of theorem \ref{th} $W_{\beta
,B}^{\alpha,A}(p,q)=W_{\beta}^{\alpha}(p,q)$ for the special component where
$~\alpha=\beta=1$
\[
W(p,q)=-q^{2}\int_{\underline{\theta}}~^{1}\langle\theta|\varphi
_{-1}(0)|\underline{\theta}\rangle_{\underline{\alpha}}~^{\underline{\alpha}%
}\langle\underline{\theta}|\varphi_{1}(0)|\theta\rangle_{1}~(2\pi)^{2}%
\delta^{(2)}\left(  p+q-R\right)
\]
Then using (\ref{w}) and (\ref{Fphipm}) the function $w(u)$ for $Z(3)$ in
(\ref{g3}) is%
\begin{equation}
w(u)=-q^{2}\frac{\left(  e^{-u}+1\right)  ^{2}}{\left(  1+\left(  2+2\cosh
u\right)  /(-q^{2}/m^{2})\right)  ^{2}}\frac{1}{\cosh^{2}\frac{1}{2}u}%
\frac{\bar{F}(-u)\bar{F}(u)}{\left(  F\left(  i\pi\right)  \right)  ^{2}}
\label{wZN}%
\end{equation}
because by (\ref{XMZN}) $\mathbf{M}(\underline{\hat{\theta}})_{\bar{1}%
1,1}^{1,\bar{1}1}~\mathbf{M}(\underline{\hat{\theta}})_{1,\bar{1}1}^{\bar
{1}1,1}=\mathbf{M}(\underline{\hat{\theta}})_{1\bar{1},1}^{1,1\bar{1}%
}~\mathbf{M}(\underline{\hat{\theta}})_{1,1\bar{1}}^{1\bar{1},1}=\left(
e^{-\hat{\theta}_{1}}+e^{-\hat{\theta}_{2}}\right)  ^{2}=\left(
e^{-u}+1\right)  ^{2}$ (see (\ref{theta0}) and Lemma \ref{las}). The function
$g(q^{2})$ for the $Z(3)$-model is given by (\ref{g3}), (\ref{wZN}) and is plotted
in Figure \ref{fzn}.
\begin{figure}
[h]
\begin{center}
\includegraphics[
height=2.2475in,
width=3.7502in
]%
{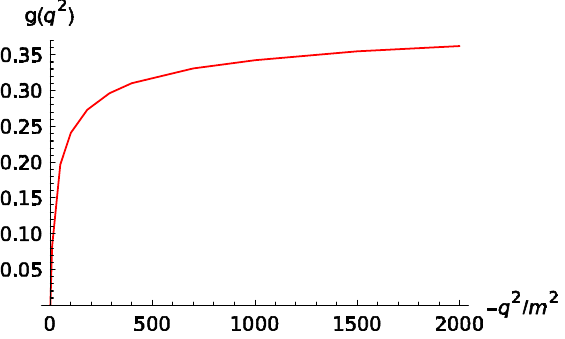}%
\caption{Plot of $g(q^{2})$ versus $-q^{2}/m^{2}$ for $Z(3)$}%
\label{fzn}%
\end{center}
\end{figure}

\section{Conclusion}

In this work, we have investigated the structure functions of two groups of
integrable quantum field theories in (1+1) dimensions. The first group
corresponds to S-matrices that are rational solutions of the Yang-Baxter
equation as for the $O(N)~\sigma$-model and the $SU(N)$ chiral Gross-Neveu
model. The second group corresponds to trigonometric solutions as for the
Sine-Gordon model and models without backward scattering as the Sinh-Gordon
and the $Z(N)$-Ising model in the scaling limit. The structure functions of
the first group, which describe asymptotically free quantum field theories,
exhibit a universal behavior $x^{-1}\ln^{-2}x$ at small Bjorken variable. This
result confirms the conjecture proposed by Balog and Weisz
\cite{Balog1,Balog2}. For the second group of S-matrices, the structure
functions show power behavior, in particular, proportional to $x.$ The
Sine-Gordon model shows a power-law dependence $x^{2/\nu-1}$ in the regime of
the coupling where there are only solitons. This behavior agrees well with the
experimental data from HERA and ZEUS (see \cite{BKS}) for certain values of
$q^{2}$ and small Bjorken $x$. This property makes the Sine-Gordon model
particularly interesting. Based on these observations, one can extend the
conjecture of Balog and Weisz \cite{Balog1,Balog2}, who suggested that the
universal behavior of structure functions at small $x$ may also appear in
four-dimensional QCD. The Sine-Gordon model therefore provides a concrete
two-dimensional candidate for the longitudinal sector of QCD (see \cite{BKS}
for further discussions).

\section{Acknowledgements}

H.B. acknowledges Armenian HESC grant 21AG-1C024 for financial support. H.B.
and M.K.  thank A. Sedrakyan for collaboration. A.F. acknowledges support from
CNPq (Conselho Nacional de
Desenvolvimento Cient\'{\i}fico e Tecnol\'{o}gico). M.K. thanks Felix von
Oppen for hospitality at the Institut f\"{u}r Theoretische Physik, Freie
Universit\"{a}t Berlin.

\appendix

\section*{Appendix}

\addcontentsline{toc}{part}{Appendix}

\renewcommand{\theequation}{\mbox{\Alph{section}.\arabic{equation}}} \setcounter{equation}{0}

\section{New variables}

\label{a0}

We use the transformation $\underline{\theta}\rightarrow\underline{u},\Lambda
$, defined by%
\begin{align*}
u_{i}  &  =\theta_{i}-\theta_{i+1},~1\leq i\leq r-1\\
R^{\pm}/m  &  =\mu(\underline{u})e^{\pm\Lambda}=e^{\pm\theta_{1}}+\dots
+e^{\pm\theta_{r}}%
\end{align*}
with
\[
\mu^{2}(\underline{u})=R^{2}/m^{2}=r+2\sum_{1\leq i<j\leq r}
\cosh\left(u_{i}+\dots+u_{j-1}\right)  .
\]
The Jacobian satisfies $\det\left(  \partial\left(  \underline{u}%
,\Lambda\right)  /\partial\left(  \underline{\theta}\right)  \right)  =1$. For
details, in particular, for the inverse transformation $\underline{\theta
}\left(  \underline{u},\Lambda\right)  $ see \cite{Balog2}%
\begin{equation}
\theta_{i}\left(  \underline{u},\Lambda\right)  =u_{i}+\dots+u_{r-1}%
+\Lambda-v^{+}(\underline{u})+v^{-}(\underline{u}) \label{t}%
\end{equation}
where%
\begin{equation}
v^{\pm}(\underline{u})=\tfrac{1}{2}\ln\left(
\sum_{i=1}^{r-1}
\,e^{\pm\left(  u_{i}+\dots+u_{r-1}\right)  }+1\right)  \label{vpm}%
\end{equation}
with%
\[
v^{+}(\underline{u})+v^{-}(\underline{u})=\tfrac{1}{2}\ln\mu^{2}(\underline
{u}).
\]
Finally we obtain for the function $W_{r}(p,q)$ of (\ref{WW})%
\[
W_{r}(p,q)=\frac{1}{4\left(  4\pi\right)  ^{r-2}}\int d\underline{u}%
\int_{-\infty}^{\infty}d\Lambda~J_{r}\left(  \underline{\theta}(\underline
{u},\Lambda),\theta\right)  ~\delta^{(2)}\left(  p+q-R\right)
\]
where $\int d\underline{u}=\int_{0}^{\infty}du_{1}\dots\int_{0}^{\infty
}du_{r-1}$.

\paragraph{$\Lambda$ and $u_{r-1}$ integrations:}

We may rewrite $\delta^{(2)}\left(  p+q-R\right)  $ as%
\[
\delta^{(2)}\left(  p+q-R\right)  =\frac{2}{m^{2}}\delta\left(  \mu
^{2}(\underline{u})-\left(  p+q\right)  ^{2}/m^{2}\right)  \delta\left(
\Lambda-\Lambda_{0}\right)
\]
with $\Lambda_{0}=\frac{1}{2}\ln\frac{R^{+}}{R^{-}}=\frac{1}{2}\ln
\frac{\left(  p+q\right)  ^{+}}{\left(  p+q\right)  ^{-}}$, note that
$\frac{\partial\left(  R^{0},R^{1}\right)  }{\partial\left(  \mu^{2}%
,\Lambda\right)  }=\frac{1}{2}m^{2}$. Then we perform the $\Lambda$ and the
$u_{r-1}$ integrations%
\begin{align*}
W_{r}(p,q)  &  =\frac{1}{\left(  4\pi\right)  ^{r-2}}\frac{1}{2m^{2}}\int
d\underline{u}~J_{r}\left(  \underline{\theta}(\underline{u},\Lambda
_{0})\right)  \delta\left(  \mu^{2}(\underline{u})-\left(  p+q\right)
^{2}/m^{2}\right) \\
&  =\frac{1}{\left(  4\pi\right)  ^{r-2}}\frac{1}{4m^{2}}\int d\underline
{u}_{r-1}~\frac{1}{\sum_{i=1}^{r-1}\sinh\left(  u_{i}+\dots+u_{r-1}%
^{(0)}\right)  }J_{r}\left(  \underline{\theta}(\underline{u}^{(0)}%
,\Lambda_{0})\right)
\end{align*}
where $\underline{u}^{(0)}=u_{1},\dots,u_{r-2},u_{r-1}^{(0)},$ and
$u_{r-1}^{(0)}$ satisfies
\begin{equation}
\mu^{2}(\underline{u}^{(0)})=\left(  p+q\right)  ^{2}/m^{2}. \label{u0}%
\end{equation}
It has been used that
\[
\frac{\partial}{\partial u_{r-1}}\mu^{2}(\underline{u})=2\sum_{i=1}^{r-1}%
\sinh\left(  u_{i}+\dots+u_{r-1}\right)  \,.
\]

\paragraph{$\Lambda_{0}$ and $u_{r-1}^{(0)}$ in terms of $x$ and $\kappa$:}

We take $p=\binom{m}{0}$, then we have%
\begin{align*}
\left(  p+q\right)  ^{\pm}/m  &  =\left(  \left(  m+q^{0}\right)  \pm
q^{1}\right)  /m=1+2\kappa^{2}\frac{1}{x}\pm2\kappa\sqrt{\kappa^{2}+x^{2}%
}\frac{1}{x}\\
&  =1+2\kappa^{2}\frac{1}{x}\left(  1\pm\sqrt{1+x^{2}/\kappa^{2}}\right)
\end{align*}
because%
\[
q^{0}/m=pq/m^{2}=2\kappa^{2}\frac{1}{x},~q^{1}/m=\sqrt{\left(  q^{0}\right)
^{2}-q^{2}}/m=2\kappa^{2}\frac{1}{x}\left(  \sqrt{1+x^{2}/\kappa^{2}}\right)
.
\]
Therefore%
\begin{equation}
\Lambda_{0}=\frac{1}{2}\ln\frac{\left(  p+q\right)  ^{+}}{\left(  p+q\right)
^{-}}=\frac{1}{2}\ln\frac{1+2\kappa^{2}\frac{1}{x}\left(  1+\sqrt
{1+x^{2}/\kappa^{2}}\right)  }{1+2\kappa^{2}\frac{1}{x}\left(  1-\sqrt
{1+x^{2}/\kappa^{2}}\right)  } \label{l0}%
\end{equation}
and $u_{r-1}^{(0)}$ is a solution of (\ref{u0})
\begin{align}
\mu^{2}(\underline{u}^{(0)})  &  =1+\mu_{r-1}^{2}(\underline{u}_{r-1})+2%
{\textstyle\sum\nolimits_{1\leq i<r}}
\cosh\left(  u_{i}+\dots+u_{r-1}^{(0)}\right) \nonumber\\
&  =\left(  p+q\right)  ^{2}/m^{2}=1+2pq/m^{2}+q^{2}/m^{2}=1+\left(  \frac
{1}{x}-1\right)  4\kappa^{2} \label{u01}%
\end{align}
with $\mu_{r-1}^{2}(\underline{u}_{r-1})=r-1+2%
{\textstyle\sum\nolimits_{1\leq i<j\leq r-1}}
\cosh\left(  u_{i}+\dots+u_{j-2}\right)  $ analogously to (\ref{mu}).

\paragraph{ Small $x$:}

\begin{lemma}
\label{lu-1}After the $\Lambda$ and the $u_{r-1}$ integrations the rapidities
of eq. (\ref{t}) $\underline{\theta}(\underline{u}^{(0)},\Lambda_{0})$ satisfy
for $p=\binom{m}{0}$ and for small $x$
\begin{align*}
\theta_{i}(\underline{u}^{(0)},\Lambda_{0})  &  =u_{i}+\dots+u_{r-2}%
+W,~i=1,\dots,r-1\\
\theta_{r}\left(  \underline{u}^{(0)},\Lambda_{0}\right)   &  =\epsilon
\end{align*}
with%
\begin{align*}
W  &  =u_{r-1}^{(0)}+\Lambda_{0}-v_{+}+v_{-}=-\ln x+O(1)\\
\epsilon &  =\Lambda_{0}-v_{+}+v_{-}=x\left(  1+\frac{1}{4\kappa^{2}}\mu
_{r-1}^{2}\right)  +O(x^{2}).
\end{align*}
From (\ref{Wr}) we obtain for small $x$%
\[
W_{r}(p,q)=\frac{1}{\left(  4\pi\right)  ^{r-2}}\frac{x}{4\kappa^{2}}\frac
{1}{2m^{2}}\int d\underline{u}_{r-1}~J_{r}\left(  \underline{\theta
}(\underline{u}^{(0)},\Lambda_{0}),0\right)  +O(x^{2})
\]

\end{lemma}

\begin{proof}
By (\ref{u01})%
\[
\mu^{2}(\underline{u}^{(0)})=\frac{4\kappa^{2}}{x}+1-4\kappa^{2}\,,
\]
and further for $x$ small $u_{r-1}^{(0)}\rightarrow\infty$ and
\begin{align*}
\mu^{2}(\underline{u}^{(0)}) &  =1+\mu_{r-1}^{2}(\underline{u}_{r-1})+%
{\textstyle\sum\nolimits_{i=1}^{r-1}}
e^{u_{i}+\dots+u_{r-1}^{(0)}}+O(e^{-u_{r-1}^{(0)}})\\
&  \Rightarrow~u_{r-1}^{(0)}=\ln\frac{4\kappa^{2}}{x}-2v_{+}^{(r-1)}+O(x)
\end{align*}
with analogously to (\ref{vpm}) $v_{(r-1)}^{\pm}=\frac{1}{2}\ln\left(
{\textstyle\sum\nolimits_{i=1}^{r-2}}
e^{\pm\left(u_{i}+\dots+u_{r-2}\right)}+1\right)$. Equation (\ref{l0})
implies%
\begin{equation}
\Lambda_{0}=\frac{1}{2}\ln\frac{4\kappa^{2}}{x}+\left(  \frac{1}{8\kappa^{2}%
}+\frac{1}{2}\right)  x+O\left(  x^{2}\right)  \,.
\end{equation}
Similarly we get
\begin{equation}
v_{-}\left(  \underline{u}^{(0)}\right)  =\frac{x}{8\kappa^{2}}\mu_{r-1}%
^{2}+O(x^{2}).\nonumber
\end{equation}
Finally
\[
W=u_{r-1}^{(0)}+\Lambda_{0}-\frac{1}{2}\ln\mu^{2}(\underline{u})+2v_{-}=-\ln
x+O(1)
\]
and%
\[
\epsilon=\Lambda_{0}-\frac{1}{2}\ln\mu^{2}(\underline{u}^{(0)})+2v_{-}\left(
\underline{u}^{(0)}\right)  =x\left(  1+\frac{1}{4\kappa^{2}}\mu_{r-1}%
^{2}\right)+O(x^{2}).
\]
To get $W_{r}(p,q)$ for small $x$ we use that (\ref{u01}) implies
\[\left(
2\sum_{i=1}^{r-1}\sinh\left(  u_{i}+\dots+u_{r-1}^{(0)}\right)  \right)
^{-1}=\frac{x}{4\kappa^{2}}+O(x^{2}).
\]
\end{proof}

\section{O(N) $\sigma$-model}

\label{apON}

The two particle S-matrix of the $O(N)~\sigma$-model $\mathbf{S}%
(\theta)=b(\theta)\mathbf{1}+c(\theta)\mathbf{P}+d(\theta)\mathbf{K}$ reads in
terms of the components as \cite{ZZ3}%
\begin{equation}
S_{\alpha\beta}^{\delta\gamma}(\theta)=b(\theta)\delta_{\alpha}^{\gamma}%
\delta_{\beta}^{\delta}+c(\theta)\delta_{\alpha}^{\delta}\delta_{\beta
}^{\gamma}+d(\theta)\delta^{\delta\gamma}\delta_{\alpha\beta} \label{son}%
\end{equation}
where $\alpha,\beta,\gamma,\delta=1,\dots,N$ denote the fundamental
$O(N)$-vector particles. The amplitudes are given by
\begin{align*}
a(\theta)  &  =b(\theta)+c(\theta)=-\exp\left(  -2\int_{0}^{\infty}\frac
{dt}{t}\frac{e^{-t\nu}+e^{-t}}{1+e^{-t}}\sinh t\frac{\theta}{i\pi}\right)
~,~~\nu=\frac{2}{N-2}\\
b(\theta)  &  =a(\theta)\frac{\theta}{\theta-i\pi\nu},~c(\theta)=-\frac
{i\pi\nu}{\theta}b(\theta),~d(\theta)=\frac{i\pi\nu}{\theta-i\pi}b(\theta)\,.
\end{align*}
The asymptotic behavior for $\theta\rightarrow\pm\infty$
\begin{equation}
\mathbf{S}(\theta)=\mathbf{1}-\frac{i\pi\nu}{\theta}\left(  \mathbf{P}%
-\mathbf{K}\right)  +O(\theta^{-2}) \label{sona}%
\end{equation}
is obtained by the formula for $x\rightarrow\infty$%
\begin{equation}
\int\limits_{0}^{\infty}\frac{1}{t}g(t)\left(  \sin tx\right)  dt=g\left(
0\right)  \frac{1}{2}\pi+g^{\prime}\left(  0\right)  \frac{1}{x}-\frac{1}%
{3}g^{(3)}\left(  0\right)  \frac{1}{x^{3}}+O(x^{-5}). \label{as}%
\end{equation}
This implies the formulae (\ref{ta}) and (\ref{tab}) with \cite{dVK,BFK5}%
\begin{equation}
X(W)=\mathbf{-}i\pi\nu\frac{1}{W}~\text{and }M_{\alpha\beta}^{\delta\gamma
}=\left(  \mathbf{P-K}\right)  _{\alpha\beta}^{\delta\gamma}=\delta_{\alpha
}^{\delta}\delta_{\beta}^{\gamma}-\delta^{\delta\gamma}\delta_{\alpha\beta}\,.
\label{XMON}%
\end{equation}
The minimal form factor function corresponding to $S_{-}(\theta)=b(\theta
)-c(\theta)$ is \cite{KW}
\begin{equation}
F_{-}(\theta)=\exp\left(  \int_{0}^{\infty}\frac{dt}{t\sinh t}\left(
\frac{e^{-t\nu}-1}{1+e^{t}}\right)  \left(  1-\cosh t(1-\theta/(i\pi))\right)
\right)  . \label{FON}%
\end{equation}

\noindent{\bf The proof Lemma \ref{las} for $O(N)$:}

\begin{proof}
We set $\underline{\hat{\alpha}}=\alpha_{1},\dots,\alpha_{r-1}$. For
$p=\binom{m}{0}$ or $\theta=0$ the form factor equations (\ref{fii}) and
(\ref{fiii}) imply for $\theta_{r}\rightarrow0$%
\begin{align*}
^{\alpha}\langle\theta|\mathcal{O}(0)|\underline{\theta}\rangle_{\underline
{\alpha}}  & =F_{\underline{\hat{\alpha}}\alpha_{r}\alpha^{\prime}%
}^{\mathcal{O}}(\underline{\theta}_{r-1},\theta_{r},-i\pi)\mathbf{C}%
^{\alpha^{\prime}\alpha}\\
& \overset{\theta_{r}\rightarrow0}{\rightarrow}\frac{2i}{\theta_{r}%
}F_{\underline{\hat{\alpha}}^{\prime}}^{\mathcal{O}}(\underline{\theta}%
_{r-1})\left(  \mathbf{1}-S(\theta_{1r})\dots S(\theta_{r-1r})\right)
_{\underline{\hat{\alpha}}\alpha_{r}}^{\alpha\underline{\hat{\alpha}}^{\prime
}}%
\end{align*}
where $\underline{\theta}_{r-1}=(\theta_{1},\dots,\theta_{r-1})$.

 Further we
set $\underline{\theta}_{r-1}=\underline{\hat{\theta}}+W=(\hat{\theta}%
_{1}+W,\dots,\hat{\theta}_{r-1}+W)$. If in addition $W\rightarrow\infty$, for
$\mathcal{O=}$ Lorentz scalar we get%
\begin{multline*}
^{\alpha}\langle\theta|\mathcal{O}(0)|\underline{\theta}\rangle_{\underline
{\alpha}}\overset{\theta_{r}\rightarrow0}{\rightarrow}\overset{W\rightarrow
\infty}{\rightarrow}\frac{2i}{\theta_{r}}\langle0|\mathcal{O}(0)|\underline
{\hat{\theta}}\rangle_{\underline{\hat{\alpha}}^{\prime}}\left(
\mathbf{1}-\left(  \mathbf{1}+X(W)\mathbf{M}\right)  \right)  _{\underline
{\hat{\alpha}}\alpha_{r}}^{\alpha\underline{\hat{\alpha}}^{\prime}}\\
=-\frac{2i}{\theta_{r}}X(W)\langle0|\mathcal{O}(0)|\underline{\hat{\theta}%
}\rangle_{\underline{\hat{\alpha}}^{\prime}}\mathbf{M}_{\underline{\hat
{\alpha}}\alpha_{r}}^{\alpha\underline{\hat{\alpha}}^{\prime}}%
\end{multline*}
because of (\ref{ta}). It has been used, that the form factor equation
(\ref{fv}) implies for a Lorentz scalar $\langle0|\mathcal{O}|\underline
{\hat{\theta}}+W\rangle_{\underline{\hat{\alpha}}}=\langle0|\mathcal{O}%
|\underline{\hat{\theta}}\rangle_{\underline{\hat{\alpha}}}$.

Similarly we have
\begin{multline*}
^{\underline{\alpha}}\langle\underline{\theta}|\mathcal{O}^{\prime}%
(0)|\theta\rangle_{\beta}\overset{\theta_{r}\rightarrow0}{\rightarrow}%
\overset{W\rightarrow\infty}{\rightarrow}\frac{2i}{-\theta_{r}}\left(
\mathbf{1}-\left(  \mathbf{1}+X^{\ast}(W)\mathbf{M}\right)  \right)
_{\beta\underline{\hat{\alpha}}^{\prime}}^{\underline{\hat{\alpha}}\alpha_{r}%
}~^{\underline{\hat{\alpha}}^{\prime}}\langle\underline{\hat{\theta}%
}|\mathcal{O}^{\prime}(0)|0\rangle\\
=-\frac{2i}{-\theta_{r}}X^{\ast}(W)\mathbf{M}_{\beta\underline{\hat{\alpha}%
}^{\prime}}^{\underline{\hat{\alpha}}\alpha_{r}}~^{\underline{\hat{\alpha}%
}^{\prime}}\langle\underline{\hat{\theta}}|\mathcal{O}^{\prime}(0)|0\rangle
\end{multline*}
and (\ref{XX}).
\end{proof}

\section{The $SU(N)$ chiral Gross-Neveu model}

\label{apSUN}

The S-matrix of two particles $S(\theta)=\mathbf{1}b(\theta)+\mathbf{P}%
c(\theta)$ and of a particle and an anti-particle $\bar{S}(\theta
)=(-1)^{N-1}\left(  \mathbf{1}b(\pi i-\theta)+\mathbf{K}c(\pi i-\theta
)\right)  $ are in terms of the components \cite{BKKW,BFK1,BFK2}
\begin{align}
S_{\alpha\beta}^{\delta\gamma}(\theta)  &  =\delta_{\alpha}^{\gamma}%
\delta_{\beta}^{\delta}b(\theta)+\delta_{\alpha}^{\delta}\delta_{\beta
}^{\gamma}c(\theta)\label{sun}\\
\bar{S}_{\bar{\alpha}\beta}^{\delta\bar{\gamma}}(\theta)  &  =(-1)^{N-1}%
\left(  \delta_{\bar{\alpha}}^{\bar{\gamma}}\delta_{\beta}^{\delta}\,b(\pi
i-\theta)+\mathbf{C}^{\delta\bar{\gamma}}\mathbf{C}_{\bar{\alpha}\beta}\,c(\pi
i-\theta)\right)  \label{sunb}%
\end{align}
where $\alpha,\beta,\gamma,\delta=1,\dots,N$ denote the fundamental
$SU(N)$-vector particles. Du to Swieca et al. \cite{KuS,KKS} the particles
satisfy an anyonic statistics and the antiparticle $\bar{\alpha}$ of a
particle $\alpha$ is to be identified with the bound state of $N-1$ particles
$\bar{\alpha}=(\alpha_{1},\dots,\alpha_{N-1})$ with $\alpha_{i}\neq\alpha$ and
$\alpha_{1}<\dots<\alpha_{N-1}$. As a consequence, the crossing relation
(\ref{sunb}) has this unusual factor $(-1)^{N-1}$. Also the asymptotic
formulae (\ref{ta}) and (\ref{tab}) are modified. However, for charge zero
operators the results (\ref{XX}), (\ref{W1}), (\ref{f}) and (\ref{g}) hold
unchanged.
The charge conjugation matrix
in (\ref{sunb}) for the bound state $\bar{\alpha}=(\alpha_{1}\dots\alpha
_{N-1})$ is defined by (see \cite{BFK1})
\begin{align}
\mathbf{C}_{\bar{\alpha}\alpha}\,  &  =\epsilon_{\alpha_{1}\dots\alpha
_{N-1}\alpha},~~\mathbf{C}_{\alpha\bar{\alpha}}=\epsilon_{\alpha\alpha
_{1}\dots\alpha_{N-1}}\nonumber\\
\mathbf{C}^{\bar{\alpha}\alpha}\,\,  &  =(-1)^{N-1}\,\epsilon^{\alpha_{1}%
\dots\alpha_{N-1}\alpha},~~\mathbf{C}^{\alpha\bar{\alpha}}=(-1)^{N-1}%
\,\epsilon^{\alpha\alpha_{1}\dots\alpha_{N-1}} \label{c}%
\end{align}
with $\epsilon_{\alpha_{1}\dots\alpha_{N}}$ and $\epsilon^{\alpha_{1}%
\dots\alpha_{N}}$ being total anti-symmetric and $\epsilon_{1\dots N}=$
$\epsilon^{1\dots N}=1$.

Some formulae for the $SU(N)$ chiral Gross-Neveu model are more complicated
compared to the $O(N)~\sigma$-model:

\begin{enumerate}
\item There are particles with different masses (bound states).

\item The statistics of the particles is anyonic.

\item The form factor equations contain some additional phase factors.
\end{enumerate}

The definition of the momentum $R=\sum_{i=1}^{r}p_{i},~p_{i}=m_{i}\binom
{\cosh\theta_{i}}{\sinh\theta_{i}}$ changes and therefore%
\begin{align*}
\mu^{2}(\underline{u})  &  =\sum_{i=1}^{r}m_{i}^{2}/m^{2}+2%
{\textstyle\sum\nolimits_{1\leq i<j\leq r}}
m_{i}m_{j}/m^{2}\cosh\left(  u_{i}+\dots+u_{j-1}\right) \\
v^{\pm}(\underline{u})  &  =\frac{1}{2}\ln\left(
{\textstyle\sum\nolimits_{i=1}^{r}}
\,m_{i}/m\,e^{\pm\left(  u_{i}+\dots+u_{r-1}\right)  }+1\right)
\end{align*}

\paragraph{The form factor equations for $SU(N)$ are}

\begin{itemize}
\item[(i)] The Watson's equations describe the symmetry property under the
permutation of both, the variables $\theta_{i},\theta_{j}$ and the spaces
$i,j=i+1$ at the same time
\begin{equation}
F_{\dots ij\dots}^{\mathcal{O}}(\dots,\theta_{i},\theta_{j},\dots)=F_{\dots
ji\dots}^{\mathcal{O}}(\dots,\theta_{j},\theta_{i},\dots)\,S_{ij}(\theta_{ij})
\label{fiSU}%
\end{equation}
for all possible arrangements of the $\theta$'s.

\item[(ii)] The crossing relation implies a periodicity property under the
cyclic permutation of the rapidity variables and spaces for the connected part%
\begin{multline}
^{\bar{1}}\langle\,\theta_{1}\,|\,\mathcal{O}(0)\,|\,\theta_{2},\dots
,\theta_{n}\,\rangle_{2\dots n}^{{conn.}}\\
=F_{1\ldots n}^{\mathcal{O}}(\theta_{1}+i\pi,\theta_{2},\dots,\theta
_{n})\,\dot{\sigma}_{1}\mathbf{C}^{\bar{1}1}=F_{2\ldots n1}^{\mathcal{O}%
}(\theta_{2},\dots,\theta_{n},\theta_{1}-i\pi)\mathbf{C}^{1\bar{1}}
\label{fiiSU}%
\end{multline}
where $\dot{\sigma}_{\alpha}^{\mathcal{O}}$ takes into account the statistics
of the particle $\alpha$ with respect to $\mathcal{O}$. The charge conjugation
matrix $\mathbf{C}^{\bar{1}1}$ is defined by (\ref{c}).

\item[(iii)] There are poles determined by one-particle states in each
sub-channel given by a subset of particles of the state. In particular, the
function $F_{\underline{\alpha}}^{\mathcal{O}}(\underline{\theta})$ has a pole
at $\theta_{12}=i\pi$ such that
\begin{equation}
\operatorname*{Res}_{\theta_{12}=i\pi}F_{1\dots n}^{\mathcal{O}}(\theta
_{1},\dots,\theta_{n})=2i\,\mathbf{C}_{12}\,F_{3\dots n}^{\mathcal{O}}%
(\theta_{3},\dots,\theta_{n})\left(  \mathbf{1}-\dot{\sigma}_{2}S_{2n}\dots
S_{23}\right)  \,. \label{fiiiSU}%
\end{equation}

\item[(iv)] If there are also bound states in the model the function
$F_{\underline{\alpha}}^{\mathcal{O}}({\underline{\theta}})$ has additional
poles. If for instance the particles 1 and 2 form a bound state (12), there is
a pole at $\theta_{12}=i\eta,~(0<\eta<\pi)$ such that
\begin{equation}
\operatorname*{Res}_{\theta_{12}=i\eta}F_{12\dots n}^{\mathcal{O}}(\theta
_{1},\theta_{2},\dots,\theta_{n})\,=F_{(12)\dots n}^{\mathcal{O}}%
(\theta_{(12)},\dots,\theta_{n})\,\sqrt{2}\Gamma_{12}^{(12)} \label{fivSU}%
\end{equation}
where the bound state intertwiner $\Gamma_{12}^{(12)}$ and the values of
$\theta_{1},~\theta_{2},~\theta_{(12)}$ and $\eta$ are given in \cite{K1,BK}.

\item[(v)] Naturally, since we are dealing with relativistic quantum field
theories we finally have
\begin{equation}
F_{1\dots n}^{\mathcal{O}}(\theta_{1},\dots,\theta_{n})=e^{s\mu}\,F_{1\dots
n}^{\mathcal{O}}(\theta_{1}+\mu,\dots,\theta_{n}+\mu) \label{fvSU}%
\end{equation}
if the local operator transforms under Lorentz transformations as
$\mathcal{O}\rightarrow e^{s\mu}\mathcal{O}$ where $s$ is the
\textquotedblleft spin\textquotedblright\ of $\mathcal{O}$.
\end{itemize}

Here $\dot{\sigma}_{1}=\rho\sigma_{1}^{\mathcal{O}}$ is a phase factor, where
$\sigma_{1}^{\mathcal{O}}$ is the statistics factor of the operator
$\mathcal{O}(x)$ with respect to the particle $1$ and $\rho$ is a sign factor
due to the unusual crossing relation of the S-matrix which is determined as
follows. The statistics factors in (ii) and (iii) are not arbitrary,
consistency implies that both are the same and
\[
\dot{\sigma}_{1}^{\mathcal{O}}\dot{\sigma}_{\bar{1}}^{\mathcal{O}}=\left(
\dot{\sigma}_{1}^{\mathcal{O}}\right)  ^{N}=(-1)^{(N-1)Q^{\mathcal{O}}}%
\]
with the solution for a fundamental particle $\alpha$%
\begin{align}
\dot{\sigma}_{\alpha}(n)  &  =\sigma_{1}^{\mathcal{O}}\rho\label{sigmadot}\\
\sigma_{1}^{\mathcal{O}}  &  =e^{i\pi(1-1/N)Q^{\mathcal{O}}}\nonumber\\
\rho &  =(-1)^{(N-1)+(1-1/N)\left(  n-Q^{\mathcal{O}}\right)  }=\pm
1\nonumber\\
\dot{\sigma}_{\bar{\alpha}}(n)  &  =\left(  \dot{\sigma}_{\alpha}(n)\right)
^{N-1}%
\end{align}

\paragraph{Asymptotic behavior:}

The amplitudes in (\ref{sun}) and (\ref{sunb}) are given by
\cite{BKKW,BFK1,BFK2}%
\begin{align}
a(\theta)  &  =b(\theta)+c(\theta)=\frac{\Gamma\left(  -\frac{\theta}{2\pi
i}\right)  \Gamma\left(  1-\frac{1}{N}+\frac{\theta}{2\pi i}\right)  }%
{\Gamma\left(  \frac{\theta}{2\pi i}\right)  \Gamma\left(  1-\frac{1}{N}%
-\frac{\theta}{2\pi i}\right)  },~\nu=\frac{2}{N}\label{aSUN}\\
\tilde{b}(\theta)  &  =\frac{b(\theta)}{a(\theta)}=\frac{\theta}{\theta
-i\pi\nu},~\tilde{c}(\theta)=\frac{c(\theta)}{a(\theta)}=\frac{-i\pi\nu
}{\theta-i\pi\nu}=-\frac{i\pi\nu}{\theta}\tilde{b}(\theta)\nonumber
\end{align}
For $\theta\rightarrow\pm\infty$ they satisfy%
\begin{align}
&  a(\theta)\overset{\theta\rightarrow\pm\infty}{\rightarrow}e^{\mp
i\pi\left(  1-\frac{1}{2}\nu\right)  }e^{-i\pi\nu\left(  1-\frac{1}{2}%
\nu\right)  \frac{1}{\theta}}\nonumber\\
&  b(\theta)\overset{\theta\rightarrow\pm\infty}{\rightarrow}e^{\mp
i\pi\left(  1-\frac{1}{2}\nu\right)  }\left(  1+\frac{1}{2}i\nu^{2}\pi\frac
{1}{\theta}\right) \label{ba}\\
&  c(\theta)/b(\theta)\overset{\theta\rightarrow\pm\infty}{\rightarrow}%
-i\pi\nu\frac{1}{\theta}.\nonumber
\end{align}
where again the asymptotic formula (\ref{as}) was used. This implies for the
particle-particle S-matrix%
\[
S_{\alpha\beta}^{\delta\gamma}(\theta)\overset{\theta\rightarrow\pm\infty
}{\rightarrow}b(\theta)\left(  \mathbf{1}\,-\frac{i\pi\nu}{\theta}M\right)
_{\alpha\beta}^{\delta\gamma}~,~~M_{\alpha\beta}^{\delta\gamma}\mathbf{=P}%
_{\alpha\beta}^{\delta\gamma}%
\]
where here $b(\theta)$ is given by (\ref{ba}). Similarly, the particle -
anti-particle S-matrix satisfies
\[
\bar{S}_{\bar{\alpha}\beta}^{\delta\bar{\gamma}}(\theta)=\overset
{\theta\rightarrow\pm\infty}{\rightarrow}(-1)^{N-1}b(\pi i-\theta)\left(
\mathbf{1}\,-\frac{i\pi\nu}{\theta}M\right)  _{\bar{\alpha}\beta}^{\delta
\bar{\gamma}}~,~~M_{\bar{\alpha}\beta}^{\delta\bar{\gamma}}=-\mathbf{K}%
_{\bar{\alpha}\beta}^{\delta\bar{\gamma}}%
\]
where $\mathbf{K}_{\bar{\alpha}\beta}^{\delta\bar{\gamma}}%
=\mathbf{C}^{\delta\bar{\gamma}}\mathbf{C}_{\bar{\alpha}\beta}$.
Therefore the formulae (\ref{ta}) and (\ref{tab}) for $W\rightarrow\infty$ are
modified%
\begin{align}
T_{\underline{\alpha}\beta}^{\delta\underline{\gamma}}(\underline{\theta
}+W,\theta)  &  \rightarrow\left(  b^{m}(W)\left(  (-1)^{N-1}b(-W)\right)
\right)  ^{\bar{m}}\left(  \mathbf{1}+X(W)\mathbf{M}\right)  _{\underline
{\alpha}\beta}^{\delta\underline{\gamma}}\label{taSUN}\\
T_{\beta\underline{\alpha}}^{\underline{\gamma}\delta}(\theta,\underline
{\theta}+W)  &  \rightarrow\left(  b^{m}(-W)\left(  (-1)^{N-1}b(W)\right)
\right)  ^{\bar{m}}\left(  \mathbf{1}+X^{\ast}(W)\mathbf{M}\right)
_{\beta\underline{\alpha}}^{\underline{\gamma}\delta} \label{tabSUN}%
\end{align}
where $\underline{\alpha}$ denote $m$ particles and $\bar{m}$ anti-particles
and $X(W)$ is given by (\ref{XONSUN}). If $\beta$ and $\delta$ are particles,
then%
\[
\mathbf{M}\mathbf{=}\sum_{i=1}^{m+\bar{m}}\mathbf{1\dots}M_{i}\mathbf{\dots
1},~~M_{i}=\left\{
\begin{array}
[c]{ll}%
\mathbf{P} & \text{if }i~\text{correspond to a particle}\\
-\mathbf{K} & \text{if }i~\text{correspond to an anti particle.}%
\end{array}
\right.
\]
Again the asymptotic results imply the Lie-algebra generator matrix $M$ in the
vector representation or $\mathbf{M}$ in higher representations.

The minimal form factor function corresponding to $\bar{S}\left(
\theta\right)  $ is \cite{BFK2,BFK3}
\begin{equation}
\bar{F}\left(  \theta\right)  =\exp\int\limits_{0}^{\infty}\frac{dt}%
{t\sinh^{2}t}e^{\frac{t}{N}}\sinh\frac{t}{N}\left(  1-\cosh t\left(
1-\frac{\theta}{i\pi}\right)  \right)  . \label{FSUN}%
\end{equation}

\noindent{\bf The proof Lemma \ref{las} for $SU(N)$:}

Because of the anyonic statistics for $SU(N)$ the proof is modified as follows:

\begin{proof}
Form factor equation (\ref{fiSU})
\begin{multline*}
F_{1\dots n}^{\mathcal{O}}(\theta_{1},\dots,\theta_{n-2},\theta_{n-1}%
,\theta_{n})\\
=F_{n-1,n,1,\dots n-2}^{\mathcal{O}}(\theta_{n-1},\theta_{n},\theta_{1}%
,\dots,\theta_{n-2})\left(  S_{1,n}\dots S_{n-2,n}\right)  \left(
S_{1,n-1}\dots S_{n-2,n-1}\right)
\end{multline*}
and (\ref{fiiiSU})%
\begin{multline*}
\operatorname*{Res}_{\theta_{n-1,n}=i\pi}F_{n-1,n,1,\dots n-2}^{\mathcal{O}%
}(\theta_{n-1},\theta_{n},\theta_{1},\dots,\theta_{n-2})\\
=2i\,\mathbf{C}_{n-1,n}\,F_{1\dots n-2}^{\mathcal{O}}(\theta_{1},\dots
,\theta_{n-2})\left(  \mathbf{1}-\dot{\sigma}_{n}S_{n,n-2}\dots S_{n,1}%
\right)
\end{multline*}
imply%
\begin{multline}
\operatorname*{Res}_{\theta_{n-1,n}=i\pi}F_{1\dots n}^{\mathcal{O}}(\theta
_{1},\dots,\theta_{n-2},\theta_{n-1},\theta_{n})
=2i\,\mathbf{C}_{n-1,n}\,F_{1\dots n-2}^{\mathcal{O}}(\theta_{1},\dots
,\theta_{n-2})\\
\times\left(  (-1)^{\left(  N-1\right)  (n-2)}\mathbf{1}-\dot{\sigma
}_{n}S_{1,n-1}\dots S_{n-2,n-1}\right).\label{iii}%
\end{multline}
It has been used that the modified crossing \cite{BFK1} $S_{\alpha\bar{1}%
}^{\bar{1}\beta}(\theta)=(-1)^{N-1}S_{1\alpha}^{\beta1}(i\pi-\theta)$ and
unitarity implies $\mathbf{C}_{12}\left(  S_{32}\dots S_{n2}\right)  \left(
S_{31}\dots S_{n1}\right)  =(-1)^{\left(  N-1\right)  (n-2)}\mathbf{1}$.

The form factor equation (\ref{fiiSU}) implies $^{\alpha}\langle
\theta|\mathcal{O}|\underline{\theta}\rangle_{\underline{\alpha}}%
=\langle|\mathcal{O}|\underline{\theta},\theta-i\pi\rangle_{\underline{\alpha
}\bar{\alpha}}\mathbf{C}^{\bar{\alpha}\alpha}$. Let $\mathcal{O}$ be a
chargeless operator, $\alpha$ and $\alpha_{r}$ particles and $\underline
{\hat{\alpha}}=\alpha_{1},\dots,\alpha_{r-1}$ a chargeless state with $m$
particles and $m$ anti-particles, then the number of elementary particles of
the state $\underline{\hat{\alpha}}\alpha_{r}\bar{\alpha}$ is
$n=m+m(N-1)+1+(N-1)=N\left(  m+1\right)  $, because the antiparticles are
bound states of $N-1$ elementary particles. For $p=\binom{m}{0}$ or $\theta=0$
the form factor equation (\ref{iii}) implies
\begin{align*}
^{\alpha}\langle\theta|\mathcal{O}(0)|\underline{\theta}\rangle_{\underline
{\alpha}}  &  =F_{\underline{\hat{\alpha}}\alpha_{r}\bar{\alpha}}%
^{\mathcal{O}}(\underline{\theta}_{r-1},\theta_{r},-i\pi)\mathbf{C}%
^{\bar{\alpha}\alpha}\\
&  \overset{\theta_{r}\rightarrow0}{\rightarrow}\frac{2i}{\theta_{r}%
}F_{\underline{\hat{\alpha}}^{\prime}}^{\mathcal{O}}(\underline{\theta}%
_{r-1})\left(  (-1)^{\left(  N-1\right)  (n-2)}\mathbf{1}-\dot{\sigma}%
_{\bar{\alpha}}T\left(  \underline{\hat{\theta}}+W,0\right)  \right)
_{\underline{\hat{\alpha}}\alpha_{r}}^{\alpha\underline{\hat{\alpha}}^{\prime
}}\\
&  \overset{W\rightarrow\infty}{\rightarrow}\frac{2i}{\theta_{r}}%
F_{\underline{\hat{\alpha}}^{\prime}}^{\mathcal{O}}(\underline{\theta}%
_{r-1})\left(  \mathbf{1}-\left(  \mathbf{1}+X(W)\mathbf{M}\right)
_{\underline{\alpha}\beta}^{\delta\underline{\gamma}}\right)  _{\underline
{\hat{\alpha}}\alpha_{r}}^{\alpha\underline{\hat{\alpha}}^{\prime}}%
\end{align*}
It was used that $(-1)^{\left(  N-1\right)  (n-2)}=(-1)^{\left(  N-1\right)
N\left(  m+1\right)  }=1$ and for $m=\bar{m}$ relation (\ref{taSUN}) means%
\begin{multline*}
\dot{\sigma}_{\bar{\alpha}}T_{\underline{\alpha}\beta}^{\delta\underline
{\gamma}}(\underline{\theta}+W,\theta)\rightarrow\dot{\sigma}_{\bar{\alpha}%
}\left(  b^{m}(W)\left(  (-1)^{N-1}b(-W)\right)  \right)  ^{m}\\
\times\left(  \mathbf{1}+X(W)\mathbf{M}\right)  _{\underline{\alpha}\beta
}^{\delta\underline{\gamma}}+O(W^{-2})=\left(  \mathbf{1}+X^{\ast
}(W)\mathbf{M}\right)  _{\underline{\alpha}\beta}^{\delta\underline{\gamma}%
}+O(W^{-2})
\end{multline*}
because by (\ref{sigmadot}) $\dot{\sigma}_{\bar{\alpha}}$ is obtained for the
anti-particle $\bar{\alpha}$ and $n=N\left(  m+1\right)  $ and $Q=0$ as
\[
\dot{\sigma}_{\bar{\alpha}}=\rho^{N-1}=\rho^{-1}=(-1)^{-(N-1)-(1-1/N)\left(
N\left(  m+1\right)  \right)  }=(-1)^{-\left(  N-1\right)  m}%
\]
and
\[
\left(  b(W)\left(  (-1)^{N-1}b(-W)\right)  \right)  ^{m}\overset
{W\rightarrow\infty}{\rightarrow}(-1)^{\left(  N-1\right)  m}+O(W^{-2}).
\]
This proves (\ref{tW1}) for $SU(N)$, similarly (\ref{tW2}) and therefore
(\ref{XX}).
\end{proof}

\section{Sine-Gordon Model}

\label{apSG}

The SG-soliton S-matrix is
\begin{equation}
S(\theta)=\left(
\begin{array}
[c]{cccc}%
S_{ss}^{ss} & 0 & 0 & 0\\
0 & S_{\bar{s}s}^{s\bar{s}} & S_{s\bar{s}}^{s\bar{s}} & 0\\
0 & S_{\bar{s}s}^{\bar{s}s} & S_{s\bar{s}}^{\bar{s}s} & 0\\
0 & 0 & 0 & S_{\bar{s}\bar{s}}^{\bar{s}\bar{s}}%
\end{array}
\right)  =\left(
\begin{array}
[c]{cccc}%
a & 0 & 0 & 0\\
0 & b & c & 0\\
0 & c & b & 0\\
0 & 0 & 0 & a
\end{array}
\right)
\label{SSG}%
\end{equation}
where the amplitudes are given by \cite{Za0}%
\begin{align*}
\frac{b(\theta)}{a(\theta)}  &  =\frac{\sinh\theta/\nu}{\sinh(i\pi-\theta
)/\nu},~~\frac{c(\theta)}{a(\theta)}=\frac{\sinh i\pi/\nu}{\sinh(i\pi
-\theta)/\nu},~~\nu=\frac{\beta^{2}}{8\pi-\beta^{2}}=\frac{\pi}{\pi+2g}\\
a(\theta)  &  =\exp\left(  \int_{0}^{\infty}\frac{1}{t}\,\left(  \frac
{\sinh\frac{1}{2}(1-\nu)t}{\sinh\frac{1}{2}\nu t\,\cosh\frac{1}{2}t}\right)
\sinh t\frac{\theta}{i\pi}dt\right)
\end{align*}

\begin{lemma}
For $\theta\rightarrow\infty$ and $\nu>1$ the S-matrix amplitudes satisfies%
\begin{align*}
a(\pm\theta)  &  =e^{\pm\frac{1}{2}\pi i\left(  1-\frac{1}{\nu}\right)
}\left(  \left(  1+o(1)\right)  \left\{
\begin{array}
[c]{lll}%
\mp i\tan\frac{\pi}{\nu}e^{-2\theta/\nu} & \text{for} & \nu>2\\
\pm2i\cot\frac{1}{2}\pi\nu e^{-\theta} & \text{for} & \nu<2
\end{array}
\right.  \right) \\
b(\pm\theta)  &  =e^{\mp\frac{1}{2}\pi i\left(  1-\frac{1}{\nu}\right)
}+O(e^{-2\theta/\nu})\\
c(\pm\theta)  &  =\pm2ie^{\mp\frac{1}{2}\pi i\left(  1-\frac{1}{\nu}\right)
}\left(  \sin\frac{\pi}{\nu}\right)  e^{-\theta/\nu}+O(e^{-2\theta/\nu})
\end{align*}
Let $\underline{\gamma}=\underline{\gamma_{\bar{s}}\gamma_{s}}$ a chargeless
state with $m$ anti solitons $\underline{\gamma_{\bar{s}}}$ and $m$ solitons
$\underline{\gamma_{s}}$ then%
\begin{gather*}
T_{\underline{\delta}\gamma}^{\alpha\underline{\beta}}(\underline{\theta
}+W,\theta)\overset{W\rightarrow\infty}{\rightarrow}\left(
\mathbf{1\mathbf{+}}X\mathbf{(}W-\theta\mathbf{)M}(\underline{\theta})\right)
_{\underline{\delta}\gamma}^{\alpha\underline{\beta}}+O(e^{-2W/\nu})\\
T_{\alpha\underline{\beta}}^{\underline{\delta}\gamma}(\theta,\underline
{\theta}+W)\overset{W\rightarrow\infty}{\rightarrow}\left(  \mathbf{1+}%
X^{\ast}(W-\theta)\mathbf{M}(\underline{\theta})\right)  _{\alpha
\underline{\beta}}^{\underline{\delta}\gamma}+O(e^{-2W/\nu})
\end{gather*}
with%
\begin{align*}
X(W)  &  =2i\left(  \sin\frac{\pi}{\nu}\right)  e^{-W/\nu}\\
M_{s\bar{s}}^{s\bar{s}}(\theta)  &  =M_{\bar{s}s}^{\bar{s}s}(\theta
)=e^{-\theta/\nu}\,,~M_{\alpha\beta}^{\delta\gamma}=0~\text{else.}
\end{align*}
and therefore%
\begin{gather*}
\mathbf{M}_{\underline{\alpha}\beta}^{\alpha\underline{\gamma}}(\underline
{\theta})=\sum_{i,\,\alpha_{i}\neq\alpha}e^{-\theta_{i}/\nu}\left(  1\dots
P_{i}\dots1\right)  _{\underline{\alpha}\beta}^{\alpha\underline{\gamma}}\\%
\begin{array}
[c]{ccc}%
\left(  1\dots P_{i}\dots1\right)  & = &
\begin{array}
[c]{c}%
\unitlength5mm\begin{picture}(12,2) \put(10,2){\oval(10,2)[lb]} \put(0,0){\oval(10,2)[tr]} \put(1,0){\line(0,1){.7}} \put(4,0){\line(0,1){.7}} \put(1,1.2){\line(0,1){.7}} \put(4,1.2){\line(0,1){.7}} \put(6,0){\line(0,1){.7}} \put(9,0){\line(0,1){.7}} \put(6,1.2){\line(0,1){.7}} \put(9,1.2){\line(0,1){.7}} \end{picture}
\end{array}
\end{array}
\end{gather*}
The small $x$ behavior of the structure function in (\ref{f}) for
Sine-Gordon is given by (\ref{fSG}).
\end{lemma}

\begin{proof}
The relation $a(\theta)\overset{\theta\rightarrow\infty}{\rightarrow}%
e^{\frac{1}{2}\pi i\left(  1-\frac{1}{\nu}\right)  }$ follows from (\ref{as})
and the term $O(e^{-2\frac{\theta}{\nu}})$ is obtained by writing the integral
as\\
$\int_{0}^{\infty}\dots\sinh t\frac{\theta}{i\pi}dt=-\frac{1}{2}%
\int_{-\infty}^{\infty}\dots e^{-t\frac{\theta}{i\pi}}dt$, shifting the
contour and taking $\operatorname*{Res}\limits_{t=2\pi i/\nu}$:
\begin{align*}
&  \ln a(\theta)=\int_{0}^{\infty}\frac{1}{t}\,\left(  \frac{\sinh\frac{1}%
{2}(1-\nu)t}{\sinh\frac{1}{2}\nu t\,\cosh\frac{1}{2}t}\right)  \sinh
t\frac{\theta}{i\pi}dt\\
&  =-\int_{0}^{\infty}\frac{1}{t}\,\left(  1-\frac{1}{\nu}\right)  \sinh
t\frac{\theta}{i\pi}dt\\
&  -\frac{1}{2}\int_{-\infty}^{\infty}\frac{1}{t}\,\left(  \frac{\sinh\frac
{1}{2}(1-\nu)t}{\sinh\frac{1}{2}\nu t\,\cosh\frac{1}{2}t}+\left(  1-\frac
{1}{\nu}\right)  \right)  e^{-t\frac{\theta}{i\pi}}dt
~\overset{\theta\rightarrow\infty}{\rightarrow}~\frac{1}{2}i\pi
\left(  1-\frac{1}{\nu}\right)  \\
&  -2\pi i\left(  \operatorname*{Res}_{t=2\pi i/\nu}+\operatorname*{Res}%
_{t=i\pi}+\dots\right)  \frac{1}{2}\,\frac{1}{t}\left(  \frac{\sinh\frac{1}%
{2}(1-\nu)t}{\sinh\frac{1}{2}\nu t\,\cosh\frac{1}{2}t}+\left(  1-\frac{1}{\nu
}\right)  \right)  e^{-t\frac{\theta}{i\pi}}\\
&  =\frac{1}{2}i\pi\left(  1-\frac{1}{\nu}\right)  -i\tan\frac{\pi}{\nu
}e^{-\frac{2}{\nu}\theta}+2i\cot\frac{1}{2}\nu\pi e^{-\theta}+\dots
\end{align*}
Unitarity $a(-\theta)a(\theta)=1$, crossing $a(i\pi-\theta)=b(\theta
),~c(i\pi-\theta)=c(\theta)$ and $c(\theta)/a(\theta)=\sinh
(i\pi/\nu)/\sinh(i\pi-\theta)/\nu$ prove the other relations.

Let $\alpha=\gamma=s$ soliton and $\underline{\gamma}=\underline{\gamma
_{\bar{s}}\gamma_{s}}$ a chargeless state with $m$ anti-solitons
$\underline{\gamma_{\bar{s}}}$ and $m$ solitons $\underline{\gamma_{s}}$ then
\begin{equation}
T_{\underline{\delta}s}^{s\underline{\beta}}(\underline{\theta}+W,\theta
)\rightarrow\mathbf{1}_{\underline{\delta}}^{\underline{\beta}}+O(e^{-2W/\nu})
\label{ta1}%
\end{equation}
because $%
{\textstyle\prod\nolimits_{j=1}^{m}}
b(\theta_{j}+W-\theta)%
{\textstyle\prod\nolimits_{j=m+1}^{2m}}
a(\theta_{j}+W-\theta)\overset{W\rightarrow\infty}{\rightarrow}=1$
and other terms with $c$ can only appear in pairs
as $c(\theta_{j}+W-\theta)c(\theta_{k}+W-\theta)=O(e^{-2W/\nu})$.

Let $\alpha=s,~\gamma=\bar{s}~$ and $\underline{\gamma}=\underline
{\gamma_{\bar{s}}\gamma_{s}}$ a chargeless state with $m$ anti solitons
$\underline{\gamma_{\bar{s}}}$ and $m$ solitons $\underline{\gamma_{s}}$ then
\begin{multline}
T_{\underline{\alpha}s}^{s\underline{\gamma}}(\underline{\theta}%
+W,\theta)\overset{W\rightarrow\infty}{\rightarrow}\left(
{\textstyle\prod\nolimits_{j=1}^{m}}
b(\theta_{j}+W-\theta)\right)  \\
\times\left(  \sum_{i=1}^{m}\frac{c(\theta_{i}+W-\theta)}{b(\theta
_{i}+W-\theta)}\left(  1\dots P_{i}\dots1\right)  _{\underline{\alpha}%
s}^{s\underline{\gamma}}\right)  \left(  \prod_{j=m+1}^{2m}a(\theta
_{j}+W-\theta)\right)  +O(e^{-2W/\nu})\\
\overset{W\rightarrow\infty}{\rightarrow}2i\left(  \sin\frac{\pi}{\nu}\right)
e^{-\left(  W-\theta\right)  /\nu}\sum_{i=1}^{m}e^{-\theta_{i}/\nu}\left(
1\dots P_{i}\dots1\right)  _{\underline{\alpha}s}^{s\underline{\gamma}%
}+O(e^{-2W/\nu})\label{ta2}%
\end{multline}
because $c(\theta)/b(\theta)\overset{\theta\rightarrow\infty}{\rightarrow
}2i\left(  \sin\frac{\pi}{\nu}\right)  e^{-\theta/\nu}+O(e^{-2\theta/\nu})$.
The results (\ref{ta1}) and (\ref{ta2}) prove (\ref{ta}). Similarly one proves
(\ref{tab}).
\end{proof}

The minimal form factor function corresponding to $S_{-}(\theta)=b(\theta
)-c(\theta)$ is \cite{KW}
\begin{align}
F_{-}(\theta)  &  =\frac{\cosh\frac{1}{2}(i\pi-\theta)}{\cosh\frac{1}{2}%
(i\pi-\theta)/\nu}\,F(\theta)\nonumber\\
F(\theta)  &  =\exp\int_{0}^{\infty}\frac{dt}{t}\frac{\sinh\frac{1}{2}%
(1-\nu)t}{\sinh\frac{1}{2}\nu t\,\cosh\frac{1}{2}t}\frac{1-\cosh
t(1-\theta/(i\pi))}{2\sinh t} \label{FSG}%
\end{align}

The proof of Lemma \ref{las} for Sine-Gordon is similar to to that one for
$O(N)$ in Appendix \ref{apON}.

\section{Sinh-Gordon Model}

\label{apsinhG}

The amplitude $S(\theta)$ of (\ref{SSH}) behaves for $\theta\rightarrow
\pm\infty$ as
\[
S(\theta)\rightarrow1\pm4e^{\mp\theta}i\sin\pi\nu
\]
and the monodromy matrix for $W\rightarrow\infty$%
\[
T(\underline{\theta}+W,\theta)\overset{W\rightarrow\infty}{\rightarrow
}1+X(\underline{\theta}-\theta,W),~X(\underline{\theta},W)=e^{-W}4i\sin\pi
\nu\sum_{i=1}^{n}e^{-\theta_{i}}%
\]
Therefore%
\begin{equation}
X(W)=\left(  4i\sin\pi\nu\right)  e^{-W},~M(\underline{\theta})=\sum_{i=1}%
^{n}e^{-\theta_{i}} \label{Msh}%
\end{equation}

The proof of Lemma \ref{las} for Sinh-Gordon is similar to to that one for
$O(N)$ in Appendix \ref{apON}.

\section{Z(N) Model}

\label{apZN}

The amplitudes of (\ref{SZN}) behave for $\theta\rightarrow\pm\infty$ as
\[
S(\theta)\overset{\theta\rightarrow\pm\infty}{\rightarrow}e^{\pm i\eta}\left(
1\pm e^{\mp\theta}2i\sin\eta\right)  ,~\bar{S}(\theta)\overset{\theta
\rightarrow\pm\infty}{\rightarrow}e^{\mp i\eta}\left(  1\pm e^{\mp\theta
}2i\sin\eta\right)
\]
and the mnodromy-matrices for $W\rightarrow\infty$%
\[
T_{\bar{1}1,1}^{1,\bar{1}1}(\underline{\theta}+W,0)\overset{W\rightarrow
\infty}{\rightarrow}1+X(W)M(\underline{\theta}),~T_{1,\bar{1}1}^{\bar{1}%
1,1}(0,\underline{\theta}+W)\overset{W\rightarrow\infty}{\rightarrow}%
1+X^{\ast}(W)M(\underline{\theta})
\]
with%
\begin{equation}
X(W)=\left(  2i\sin\eta\right)  e^{-W},~M(\underline{\theta})=\sum_{i=1}%
^{n}e^{-\theta_{i}} \label{XMZN}%
\end{equation}
The minimal form factor function corresponding to $\bar{S}(\theta)$ is%
\begin{equation}
\bar{F}(\theta)=\exp\int_{0}^{\infty}\frac{dt}{t}\,\frac{\sinh t\frac{2}{N}%
}{\sinh^{2}t}\left(  1-\cosh t\left(  1-\frac{\theta}{i\pi}\right)  \right)
\label{FZN}%
\end{equation}

The proof of Lemma \ref{las} for $Z(N)$ is similar to to that one for $SU(N)$
in Appendix \ref{apSUN}.

\end{document}